\documentclass[11pt]{article}

\usepackage[margin=2.5cm]{geometry}
\usepackage{footnote}
\usepackage{mathtools, amsthm, amsfonts, amssymb}
\usepackage[colorlinks]{hyperref}
\hypersetup{
  linkcolor=[rgb]{0.3,0.3,0.6},
  citecolor=[rgb]{0.2, 0.6, 0.2},
  urlcolor=[rgb]{0.6, 0.2, 0.2}
}
\usepackage{xcolor}
\usepackage{enumitem}
\usepackage{thmtools}
\usepackage{thm-restate}
\usepackage{braket}
\usepackage[capitalise]{cleveref}
\usepackage{accents}
\usepackage{rotating}
\usepackage{multirow}
\newcommand{\vleq}{\rotatebox[origin=c]{-90}{$\leq$}}
\newcommand{\vvleq}{\rotatebox[origin=c]{-30}{$\leq$}}

\theoremstyle{plain}
\newtheorem{theorem}{Theorem}[section]
\newtheorem{proposition}[theorem]{Proposition}
\newtheorem{lemma}[theorem]{Lemma}
\newtheorem{corollary}[theorem]{Corollary}

\newtheorem{definition}[theorem]{Definition}
\newtheorem{remark}[theorem]{Remark}
\newtheorem{example}[theorem]{Example}

\newcommand{\sotlim}{\displaystyle \text{sot-}\lim}

\newcommand{\F}{\mathbb{F}}
\newcommand{\norm}[1]{\left\|#1\right\|}
\newcommand{\R}{\mathbb{R}}
\newcommand{\C}{\mathbb{C}}
\newcommand{\N}{\mathbb{N}}
\newcommand{\Z}{\mathbb{Z}}
\newcommand{\rk}{\mathrm{rank}}
\newcommand{\cl}{\mathrm{cl}}

\newcommand{\ran}{\mathrm{ran}}

\newcommand{\cA}{\mathcal{A}}
\newcommand{\cG}{\mathcal{G}}
\newcommand{\cH}{\mathcal{H}}
\newcommand{\cM}{\mathcal{M}}
\newcommand{\cL}{\mathcal{L}}
\newcommand{\bX}{\mathbf{X}}

\newcommand{\Qto}{\overset{\smash{q}}{\to}}
\newcommand{\QCto}{\overset{\smash{qc}}{\to}}
\newcommand{\Cto}{\overset{\smash{C^*}}{\to}}
\newcommand{\Bto}{\overset{\smash{B}}{\to}}
\newcommand{\tto}{\overset{\smash{t}}{\to}}

\DeclareMathOperator{\Hom}{Hom}

\newcommand{\eps}{\varepsilon}

\DeclareMathOperator{\Tr}{Tr}
\DeclareMathOperator{\tr}{tr}

\renewcommand{\epsilon}{\varepsilon}
\renewcommand{\eps}{\varepsilon}
\usepackage{ifdraft}
\newcommand{\xitr}{\overline{\xi}_{{tr}}}
\newcommand{\al}{\alpha}
\newcommand{\la}{\lambda}

\title{On a tracial version of Haemers bound}
\author{Li Gao\thanks{Department of Mathematics, University of Houston, Houston, TX, 77204, USA. \texttt{lgao12@uh.edu}} \and Sander Gribling\thanks{IRIF, Universit\'e de Paris, CNRS, Paris, France. \texttt{gribling@irif.fr}} \and Yinan Li\thanks{Graduate school of mathematics, Nagoya University, Nagoya, Japan. \texttt{Yinan.Li@maths.nagoya-u.ac.jp}}}
\date{\today}

\begin{document}

\maketitle
\begin{abstract}
We extend upper bounds on the quantum independence number and the quantum Shannon capacity of graphs to their counterparts in the commuting operator model. We introduce a von Neumann algebraic generalization of the fractional Haemers bound (over $\C$) and prove that the generalization upper bounds the commuting quantum independence number. We call our bound the tracial Haemers bound, and we prove that it is multiplicative with respect to the strong product. In particular, this makes it an upper bound on the Shannon capacity. The tracial Haemers bound is incomparable with the Lov\'asz theta function, another well-known upper bound on the Shannon capacity. We show that separating the tracial and fractional Haemers bounds would refute Connes' embedding conjecture.

Along the way, we prove that the tracial rank and tracial Haemers bound are elements of the (commuting quantum) asymptotic spectrum of graphs (Zuiddam, Combinatorica, 2019). We also show that the inertia bound (an upper bound on the quantum independence number) upper bounds the commuting quantum independence number.
\end{abstract}

\section{Introduction}
\subsection{Independence number from nonlocal games}
Two players, Alice and Bob, attempt to convince a referee that a given graph $G$ contains an independent set of size $d$. The referee sends two integers $i_A,i_B\in[d]$ to Alice and Bob, respectively, who then response vertices $g_A$ and $g_B$ satisfying that
\begin{itemize}
\item if $i_A=i_B$ then $g_A=g_B$;
\item if $i_A\neq i_B$ then $\{g_A,g_B\}\in E(\overline{G})$.
\end{itemize}
Alice and Bob can decide a strategy beforehand, but they cannot communicate during the game. Such a nonlocal game is known as the independent set game~\cite{manvcinska2016quantum,cj16-05}. A classical strategy corresponds to the case that Alice and Bob only use shared (classical) randomness. A perfect classical strategy enables Alice and Bob to win the independent set game for graph $G$ with certainty. It can be shown that a perfect classical strategy exists if and only if the graph $G$ has an independent set of size $d$. Thus, the largest possible $d$ for a graph $G$ where a perfect classical strategy for the independent set game exists is exactly the independence number $\alpha(G)$ of $G$.

One may also consider quantum strategies, where Alice and Bob are allowed to share quantum entanglement and perform measurements on their own systems. A perfect quantum strategy enables Alice and Bob to win the independent set game for a graph $G$ with a potentially larger $d$. The largest such $d$ leads to the definition of quantum versions of the independence number of graphs. Note that different choices of the underlying model of quantum mechanics result in different sets of quantum strategies. In the tensor-product model, the shared entangled state is modeled by a bipartite state in a tensor product Hilbert space $\cH_A\otimes \cH_B$, and the measurements performed by Alice and Bob are modeled as observables on $\cH_A$ and $\cH_B$, respectively. Define the quantum independence number of~$G$, $\alpha_q(G)$, as the largest integer $d$ for which there exists a perfect quantum strategy in the tensor-product model for the independent set game. One may also consider the commuting-operator model, where the shared entangled state is given by a state in a joint Hilbert space $\cH$, and the measurements performed by Alice and Bob are modeled by commuting observables on $\cH$. One similarly defines  the commuting quantum independence number of $G$, $\alpha_{qc}(G)$, as the largest integer $d$ for which there exists a perfect quantum strategy in the commuting-operator model for the independent set game. It is straightforward to see that for any graph $G$ we have the chain of inequalities
\[
\alpha(G)\leq\alpha_q(G)\leq\alpha_{qc}(G).
\]

\subsection{Quantum independence number: motivation and bounds}

The quantum independence number $\alpha_q(G)$ has been extensively studied in the settings of nonlocal games~\cite{6466384,chailloux_et_al:LIPIcs:2014:4807,manvcinska2016quantum,cj16-05}, combinatorics~\cite{robersonthesis,qinertia}, and zero-error quantum information theory~\cite{cubitt2010improving,Leung2012,Briet19227,briet2015entanglement}. It has been shown that perfect quantum strategies for the independent set game can always be achieved by performing projective measurements on a shared maximally entangled state (of a finite dimension)~\cite[Theorem 2.1]{manvcinska2016quantum}.
% The quantum independence number $\alpha_q(G)$ has been extensively studied in the settings of nonlocal games, combinatorics, and quantum zero-error information theory. This study was initialized in~\cite{manvcinska2016quantum}. It has been shown that perfect quantum strategies for the independent set game can always be achieved by performing projective measurements on a shared maximally entangled state (of a finite dimension).
A natural question is whether we can determine $\alpha_q(G)$ for a given graph~$G$. As a computational problem, determining whether $\alpha_q(G)\geq d$ for some integer $d$ is NP-hard~\cite{qalpha}. In fact, any algorithm for determining whether $\alpha_{q}(G)\geq d$ leads to an algorithm for determining whether a given nonlocal game admits perfect quantum strategies achieved by performing projective measurements on a shared maximally entangled state (of a finite dimension)~\cite{cj16-05}. Note that the latter family of problems includes the one which determines whether a synchronous game admits perfect quantum strategies (in the tensor-product model).

Another important reason to study the quantum independence number arises from the study of the zero-error capacity of classical channels, where the sender and the receiver are allowed to share an entangled state to assist their communication. It is known that the maximal number of zero-error messages one can send through a classical channel assisted by quantum entanglement is determined by its confusability graph~\cite{cubitt2010improving}. This correspondence leads to the definition of the entanglement-assisted independence number $\alpha_*(G)$ of a (confusability) graph $G$. From this information-theoretic point of view, the quantum independence number of the confusability graph $G$ can be equivalently interpreted as the maximal number of zero-error messages one can send through the corresponding classical channel assisted by performing projective measurements on a maximally entangled state~\cite[Section 5]{manvcinska2016quantum}. In information theory, it is also natural to consider the asymptotic zero-error rates, which leads to the definitions of the entanglement-assisted Shannon capacity $\Theta_*(G)$ and the quantum Shannon capacity $\Theta_q(G)$ of the confusability graph $G$:
\[
\Theta_q(G)=\lim_{k\to\infty}\sqrt[k]{\alpha_q(G^{\boxtimes k})}~\text{and}~\Theta_*(G)=\lim_{k\to\infty}\sqrt[k]{\alpha_*(G^{\boxtimes k})}.
\]
From the definitions, we have
\[
\alpha(G)\leq\alpha_q(G)\leq\alpha_*(G)~\text{and}~\Theta(G)\leq\Theta_q(G)\leq\Theta_*(G),
\]
where $\alpha(G)$ and $\Theta(G)$ denote the standard independence number and Shannon capacity of a graph~$G$.

Note that in the vanishing-error setting\footnote{The error probability goes to $0$ when the number of uses of the channel goes to infinity~\cite{https://doi.org/10.1002/j.1538-7305.1948.tb01338.x}.}, quantum entanglement is not helpful for enhancing the communication through classical channels~\cite{1035117}. A natural question is whether a quantum advantage can be found in the zero-error setting. In~\cite{cubitt2010improving,briet2015entanglement}, specific zero-error encoding schemes (using maximally entangled state and projective measurements) are proposed, which lead to lower bounds on $\alpha_q(G)$ (and thus also lower bounds on $\alpha_*(G)$). These lower bounds have been used to separate $\alpha(G)$ and $\alpha_q(G)$ for graphs constructed from proofs of the Kochen-Specker theorem~\cite{cubitt2010improving,6466384}. Moreover, separations of $\Theta(G)$ and $\Theta_q(G)$ can be shown by constructing geometric graphs whose Haemers bounds over certain finite fields are strictly smaller than these lower bounds~\cite{Leung2012,Briet19227,briet2015entanglement}. Here, the Haemers bound $\cH(G;\F)$ of a graph $G$ over any field $\F$ is an important classical upper bound on the Shannon capacity $\Theta(G)$~\cite{Haemers1978,haemers1979some}

On the other hand, the celebrated Lov\'asz theta function $\vartheta(G)$~\cite{lovasz1979shannon}, which is another important upper bound on the Shannon capacity $\Theta(G)$, also upper bounds $\alpha_*(G)$ and $\Theta_*(G)$~\cite{Beigi2010,duan2013}. Two central open problems in quantum zero-error information theory are to determine whether $\Theta_*(G)=\vartheta(G)$ and $\alpha_q(G)=\alpha_*(G)$ for all graphs $G$. Recently, it was observed that these two equalities cannot both hold~\cite{Li2018quantum}: $\alpha_q(G)$ and $\Theta_q(G)$ are upper bounded by the fractional Haemers bounds $\cH_f(G;\R)$ and $\cH_f(G;\C)$~\cite{blasiak2013graph,bukh2018fractional} over $\R$ and $\C$ and there exist graphs $G$ such that $\cH_f(G;\R)\leq\cH(G;\R)<\vartheta(G)$. Such an observation was made during the study of quantum asymptotic spectra of graphs, an important set of graph parameters which provide dual characterizations for the Shannon capacity and its quantum variants, initialized by Zuiddam~\cite{zuiddam2018asymptotic,phd}.

\subsection{Commuting quantum independence number: known and our results}
The central problem in the theory of nonlocal games is to understand the differences between the two different models of quantum mechanics, the (finite-dimensional) tensor-product model and the commuting-operator model. A recent breakthrough (preprint) of Ji, Natarajan, Vidick, Wright and Yuen shows that these two models are not equivalent, which refutes the celebrated Connes' embedding conjecture~\cite{JNVWY}.

Among graph-theoretic nonlocal games, the commuting quantum chromatic number and commuting quantum homomorphism have been extensively studied in the past decade. These games are defined in terms of the existence of perfect commuting quantum strategies for the graph coloring and homomorphism games, respectively~\cite{paulsen2015,PAULSEN2016,ORTIZ2016,synchronous}. In the seminal paper~\cite{PAULSEN2016}, a semidefinite programming (SDP) hierarchy is proposed which converges to the commuting quantum chromatic number. They also propose a new lower bound, the tracial rank, of the commuting quantum chromatic number. The tracial rank can be viewed as a $C^*$-algebraic generalization of the projective rank, which was introduced in~\cite{6466384} for lower bounding the quantum chromatic number.\footnote{In this paper, we evaluate the tracial rank and the projective rank on the complement of a graph $G$.} Tracial noncommutative polynomial optimization can be used to unify existing bounds and obtain SDP hierarchies of bounds on the commuting quantum independence number~\cite{gribling2018bounds}.

In this paper, we further extend recent results on quantum independence numbers to commuting quantum independence numbers. In particular, we prove that several known bounds on the quantum independence number and quantum Shannon capacity are also upper bounds on the commuting quantum independence number and commuting quantum Shannon capacity.

We first establish a characterization of the commuting quantum Shannon capacity, defined as
\[
\Theta_{qc}(G)=\lim_{k\to\infty}\sqrt[k]{\alpha_{qc}(G^{\boxtimes k})},
\]
in terms of the asymptotic spectrum of graphs~\cite{zuiddam2018asymptotic,phd}. This can be done by exhibiting desired properties of the commuting quantum homomorphism. Such characterizations indicate that, to compute the Shannon capacity and its quantum variants, it is important to find graph parameters that are additive (with respect to disjoint union), multiplicative (with respect to strong product), normalized (the value of $\overline{K}_d$ is $d$ for any $d\in\N$) and monotonic (with respect to the corresponding homomorphism). Note that almost all well-known upper bounds on the Shannon capacities satisfy these four properties, including the Lov\'asz theta function $\vartheta(G)$, fractional Haemers bound $\cH_f(G;\F)$ over any field $\F$, the projective rank $\overline{\xi}_f(G)$ and the fractional clique cover number $\overline{\chi}_f(G)$. We prove that the tracial rank $\overline{\xi}_{tr}(G)$ is another such graph parameter.

Our main contribution is a tracial version of the fractional Haemers bound (over $\C$) that upper bounds the commuting quantum independence number of graphs. We call this new bound the tracial Haemers bound. The bound is motivated by the study of the Shannon capacity (resp.~quantum Shannon capacity) of graphs, where the fractional Haemers bound over any field (resp.~over $\C$ and $\R$) and Lov\'asz theta function are incomparable elements in the asymptotic spectrum (resp.~quantum asymptotic spectrum) of graphs. We show that the tracial Haemers bound is an element of the commuting quantum asymptotic spectrum of graphs, and it is incomparable with the Lov\'asz theta function.

Our tracial Haemers bound is defined by generalizing a coordinate-free definition of the fractional Haemers bound due to Lex Schrijver~\cite[Remark after Prop.~7]{bukh2018fractional} to the von Neumann algebraic setting. When restricting to finite-dimensional von Neumann algebras, we show that it reduces to the fractional Haemers bound (over $\C$). We then prove that the tracial Haemers bound satisfies many properties of the fractional Haemers bound: it is additive, multiplicative, normalized and monotonic with respect to the commuting quantum homomorphism. This shows that the tracial Haemers bound is another element in the commuting quantum asymptotic spectrum of graphs. Thus, we provide a (potentially) better upper bound on the Shannon capacity and its quantum variants. We show the tracial Haemers bound of odd cycles equals their fractional clique cover number. This implies that the tracial Haemers bound and the Lov\'asz theta function are incomparable. To compare the tracial Haemers bound with fractional Haemers bound, we show that the existence of a graph $G$ for which $\cH_{tr}(G)<\cH_f(G;\C)$ would refute Connes' embedding conjecture.

We now have the following inequality which illustrates the relations between independence numbers, Shannon capacities and elements in the asymptotic spectrum of graphs:
\begin{center}
\begin{tabular}{ccccccccccccc}
$\alpha(G)$&$\leq$&$\alpha_q(G)$&$\leq$&$\alpha_{qc}(G)$&&&&&&&&\\
$\vleq$&&$\vleq$&&$\vleq$&&$\cH_{tr}(G)$&$\leq$&$\cH_f(G;\C)$& & &\\
$\Theta(G)$&$\leq$&$\Theta_q(G)$&$\leq$&$\Theta_{qc}(G)$&$\leq$& & \vvleq & & $\leq$&$\overline{\xi}_f(G)$&$\leq$&$\overline{\chi}_f(G)$,\\
&&&&&&$\vartheta(G)$&$\leq$&$\overline{\xi}_{tr}(G)$&&&&\\
\end{tabular}
\end{center}
where $\vartheta(G)$ is incomparable with $\cH_{tr}(G)$ and with $\cH_f(G;\C)$, and $\cH_f(G;\C)$ is not known to be comparable with $\overline{\xi}_{tr}(G)$. Moreover, our tracial Haemers bound provides a unified way to describe almost all the above elements in the asymptotic spectrum of graphs (except the Lov\'asz theta function, cf.~\cref{prop: finite dim C^* algebra,cor: abelian C*-algebra),prop: upper bound tracial rank}).

To derive the desired properties of the tracial Haemers bound, we use an infinite-dimensional version of the rank-nullity theorem. This theorem can be established using an explicit representation of the $C^*$-algebra generated by two projections~\cite{twoprojections}, which plays an important role in our proofs and has also been used in the study of tracial rank~\cite{synchronous}. As a byproduct, we prove that the inertia bound for the independence number~\cite{inertia} and the quantum independence number~\cite{WE18,qinertia}, is also an upper bound for the commuting quantum independence number. For this, we first upper bound the commuting quantum independence number by the tracial packing number, an infinite-dimensional generalization of the projective packing number introduced in~\cite{robersonthesis}. Then we upper bound the tracial packing number by the inertia bound using a characterization of the largest dimension of totally isotropic subspaces of an operator by allowing the ancillary system to be a von Neumann algebra with a tracial state. Although the inertia bound is only supermultiplicative (thus it fails to upper bound the commuting quantum Shannon capacity), it is incomparable with the Lov\'asz theta function and tracial Haemers bound. In particular, the inertia bound allows us to identify the commuting quantum independence number of graphs for which the inertia bound equals the classical independence number: famous examples include odd cycles, perfect graphs, Kneser graphs, Andrasfai graphs, the Petersen graph, the Clebsch graph and many others.
\medskip

The rest of this paper is organized as follows. In \cref{sec: pre} we list some necessary definitions and useful lemmas in graph theory and operator algebras, which will be useful for proving our results. In \cref{sec: asymptotic spectra} we establish asymptotic spectral characterizations for the $C^*$-algebraic and commuting quantum Shannon capacity of graphs, and prove that the tracial rank is an element in the commuting quantum asymptotic spectrum of graphs. In \cref{sec: tracial haemers}, we introduce the tracial Haemers bound, exhibit desired properties to show that it is an element in the commuting quantum asymptotic spectrum of graphs, and compare it with the Lov\'asz theta function and the fractional Haemers bound. In \cref{sec: inertia}, we prove the inertia bound for independence number and quantum independence number is also an upper bound for the commuting quantum independence number.

\section{Preliminaries}\label{sec: pre}
\subsection{Notation and preliminaries about graphs}\label{subsec: graph}

We work with unweighted and undirected graphs with no loops, no parallel edges and finitely many vertices. For a graph $G$, $V(G)$ denotes its vertex set and $E(G)\subseteq \{\{g,g'\}:~g,g'\in V(G)~\text{with}~g\neq g'\}$ denotes its edge set. For $g\in V(G)$, let $N_G(g)=\{g'\in V(G):~\{g,g'\}\in E(G)\}$ be the set of vertices which is adjacent to $g$. The complement graph of $G$ is the graph $\overline{G}=(V(G),E(\overline{G}))$, where $E(\overline{G})=\{\{g,g'\}: g,g'\in V(G)~\text{with}~g\neq g' \text{ and } \{g,g'\}\not\in E(G)\}$. For $d \in \N$, let $[d]:=\{1,\dots,d\}$ and $K_d$ denotes the complete graph on $d$ vertices with edge set $\{\{i,j\}:~i,j\in[d]~\text{with}~i\neq j\}$.
By convention let $K_0=(\emptyset,\emptyset)$ and $K_1=([1],\emptyset)$. Given two graphs $G$ and $H$, their {\em disjoint union}  is the graph $G\sqcup H$ with vertex set $V(G)\cup V(H)$ and edge set $E(G)\cup E(H)$. The {\em strong product} of two graphs $G$ and $H$ is the graph $G\boxtimes H$ with vertex set $V(G)\times V(H)$ and edge set
\[
\begin{split}
E(G\boxtimes H)=\big\{ \{(g,h),(g',h')\}:~&\{g,g'\}\in E(G)~\text{and}~\{h,h'\}\in E(H),\\
&~\text{or}~g=g'\in V(G)~\text{and}~\{h,h'\}\in E(H),\\
&~\text{or}~\{g,g'\}\in E(G)~\text{and}~h=h'\in V(H) \big\}.
\end{split}
\]
A {\em graph homomorphism} from $G$ to $H$ is an edge-preserving vertex map $f:V(G)\to V(H)$, i.e., a map such that $\{g,g'\}\in E(G)$ implies $\{f(g),f(g')\}\in E(H)$. The notation  $G\to H$ is used to denote that  there exists  a graph homomorphism from $G$ to $H$ and the notation $G\leq_{loc} H$ means   $\overline{G}\to \overline{H}$. The map $f:V(G)\to V(H)$ is an isomorphism if it is a bijection that preserves edges and non-edges.
Graph isomorphism is an equivalence relation on the set of graphs and we let $\{\text{Graphs}\}$ denote the set of isomorphism classes of graphs. Note that the disjoint union $\sqcup$ is an associative and commutative binary operation (addition) on $\{\text{Graphs}\}$ with identity element $K_0$, and the strong product $\boxtimes$ is an associative and commutative binary operation (multiplication) on $\{\text{Graphs}\}$ with identity element $K_1$. Thus $\cG=(\{\text{Graphs}\},\sqcup,\boxtimes,K_0,K_1)$ is a \emph{semiring}, called the semiring of graphs. The relation $\leq_{loc}$ defines a preorder on $\cG$, i.e.,~it satisfies $G\leq_{loc} G$, and $G\leq_{loc} H$, $H\leq_{loc} K$ implies $G\leq_{loc} K$ for any $G,H,K\in\cG$. The following notions of semiring homomorphism and asymptotic spectrum of graphs play a central role in this paper.

\begin{definition}[Asymptotic spectrum of graphs]\label{defring}
A {\em spectral point} $\phi$ is a $\leq_{loc}$-monotone semiring homomorphism from the semiring $\cG$ to the semiring $\R_{\geq 0}$. In particular, it satisfies the following conditions: For $G,H\in\{\emph{Graphs}\}$:
\begin{enumerate}[label=\upshape(\roman*)]
\item $\phi(\overline{K}_1)=1$,
\item $\phi(G\boxtimes H)=\phi(G)\phi(H)$,
\item $\phi(G\sqcup H)=\phi(G)+\phi(H)$,
\item $G\leq_{loc} H$ implies $\phi(G)\leq_{loc}\phi(H)$.
\end{enumerate}
The {\em asymptotic spectrum of graphs} $\bX(\cG,\leq_{loc})$ is the set of all spectral points.

\end{definition}

\begin{remark}\label{remark: zero to zero}
A semiring homomorphism is also required to map $K_0$ to $0$. This may not be well-defined for some explicit maps $\phi:\cG \to \R_{\geq 0}$. By convention we set $\phi(K_0)=0$ for all $\phi\in \bX(\cG,\leq_{loc})$.
\end{remark}

Known elements in $\bX(\cG,\leq_{loc})$ include the Lov\'asz theta function $\vartheta$~\cite{lovasz1979shannon}, the fractional Haemers bound $\cH_f(G;\F)$ over any field $\F$ ~\cite{blasiak2013graph,bukh2018fractional}, the projective rank $\overline\xi_f(G)$~\cite{manvcinska2016quantum}
and the fractional clique cover number $\overline{\chi}_f(G)$ (cf.~\cite[Eq. (67.112)]{schrijver2003combinatorial}).

The following property  of the theta function $\vartheta$, introduced in~\cite{cubitt2014bounds}, will be useful for us.

\begin{proposition}\label{prop: preorder definition for lovasz theta}
Given two graphs $G$ and $H$, we write $G\leq_B H$ if there exist an integer $d\in\N$ and  vectors $w\neq 0$, $w_{g,h}\in\C^d$ for every $g\in V(G)$ and $h\in V(H)$, satisfying the conditions:
\begin{enumerate}[label=\upshape(\roman*)]
\item $\sum_{h\in V(H)}w_{g,h}=w$ for every $g\in V(G)$,
\item $\langle w_{g,h},w_{g',h'}\rangle=0$ for all $\{g,g'\}\in E(\overline{G})$ and ($\{h,h'\}\in E(H)$ or $h=h'$),
\item $\langle w_{g,h},w_{g,h'}\rangle=0$ for all $g\in V(G)$ and $h,h'\in V(H)$ with $h\neq h'$.
\end{enumerate}
Then we have $\vartheta(G)\le \vartheta(H) $ $\Longleftrightarrow $ $G\leq_B H$.
\end{proposition}

The notion of graph homomorphism can be extended to the quantum setting via nonlocal games, as introduced in~\cite{manvcinska2016quantum}.
Let $M(d\times d',\C)$ be the set of all $d\times d'$ matrices over $\C$. Let $M(d,\C)= M(d\times d,\C)$ for short. A matrix $E\in M(d,\C)$ is a projection if $E^2=E=E^*$, where $E^*$ denotes the conjugate transpose of $E$. Let $I_d$ be the $d\times d$ identity matrix.

\begin{definition}[Quantum homomorphism~\cite{manvcinska2016quantum}]\label{def: quantum hom}
We say there is a {\em quantum homomorphism} from a graph $G$ to to a graph $H$, denoted as $G\Qto H$, if, for some $d \in \N$, there exist projections $E_{g,h}\in M(d,\C)$  for all $g\in V(G)$ and $h\in V(H)$ that satisfy
\begin{enumerate}[label=\upshape(\roman*)]
\item $\sum_{h\in V(H)}E_{g,h}=I_d$ for any $g\in V(G)$,
% \item $E_{g,h}E_{g,h'}=0$ for any $g\in V(G)$ and $h\neq h'\in V(H)$,
\item $E_{g,h}E_{g',h'}=0$ if $\{g,g'\}\in E(G)$ and ($\{h,h'\}\in E(\overline{H})$ or $h=h'\in V(H)$).
\end{enumerate}
In addition we denote $G\le_q H$ if $\overline G \Qto \overline H$.
\end{definition}

Note that a graph homomorphism $f: V(G)\to V(H)$ defines a quantum homomorphism (with $d=1$) by setting $E_{g,h}=\delta_{f(g),h}$.
Hence,  $\overline{G}\to \overline{H}$ implies $\overline{G}\Qto \overline{H}$ and thus
$G\leq_{loc} H$ implies $G\leq_q H$. $\leq_q$ is also a preorder on $\cG$. We can analogously define
\[
\bX(\cG,\leq_{q}):=\{\phi:\cG\to \R_{\geq 0}:~\phi~\text{is a}~\leq_{q}\text{-monotone semiring homomorphism}\}
\]
as the {\em quantum asymptotic spectrum of graphs}, this notion was first studied in~\cite{Li2018quantum}. By the above observations it follows that $\bX(\cG,\leq_{q}) \subseteq \bX(\cG,\leq_{loc})$. Known elements in $\bX(\cG,\leq_{q})$ include the Lov\'asz theta function $\vartheta$~\cite{lovasz1979shannon}, the fractional Haemers bound $\cH_f(G;\C)$~\cite{blasiak2013graph,bukh2018fractional}\footnote{It has been shown in~\cite{Li2018quantum} that $\cH_f(G;\C)=\cH_f(G;\R)$ for any graph $G$.} and the projective rank $\overline{\xi}_f(G)$~\cite{manvcinska2016quantum}. It remains unknown whether there are infinite many elements in $\bX(\cG,\leq_{q})$.

\subsection{A dual characterization for the (quantum) Shannon capacity of graphs}

The main motivation  to study the (quantum) asymptotic spectrum of graphs comes from its relevance to the study of (quantum versions of) independence numbers and the Shannon capacity of graphs.
Recall that  $\alpha(G)$ denotes the independence number of $G$, defined as the largest cardinality of an independent set of vertices of $G$, or, equivalently,  $\alpha(G)=\max\{d:~\overline{K}_d\leq_{loc} G\}$. Recall also that  $\overline{\chi}(G)$ is the clique cover number of $G$ (and equals the chromatic number $\chi(\overline{G})$ of $\overline{G}$), which can be equivalently defined as $\overline{\chi}(G)=\min\{d:~G\leq_{loc} \overline{K}_d\}$.
The  inequalities
$\alpha (G)\le \vartheta(G)\le \chi(\overline G)$ are known as the sandwich theorem~\cite{doi:10.1137/1.9781611970203,knuth1994sandwich}.
It follows directly from the definitions that for $\phi\in\bX(\cG,\leq_{loc})$ the sandwich theorem holds:
\[
\alpha(G)\leq\phi(G)\leq\overline{\chi}(G)\quad \text{ for any graph } G.
\]
In addition, one can  lift the sandwich theorem to asymptotic quantities. Let
\[
\Theta(G):=\sup_{k}\sqrt[k]{\alpha(G^{\boxtimes k})}=\lim_{k\to\infty}\sqrt[k]{\alpha(G^{\boxtimes k})}
\]
be the {\em Shannon capacity} of $G$~\cite{MR0089131}. It is proved by Lov\'asz~\cite{10.1016/0012-365X(75)90058-8} (see also~\cite[Theorem 67.17]{schrijver2003combinatorial}) that
\[
\overline{\chi}_f(G)=\inf_{k}\sqrt[k]{\overline{\chi}(G^{\boxtimes k})}=\lim_{k\to\infty}\sqrt[k]{\overline{\chi}(G^{\boxtimes k})}.
\]
As an application of the definitions, for any $\phi\in\bX(\cG,\leq_{loc})$, we have
 \[
\Theta(G)\leq\phi(G)\leq\overline{\chi}_f(G).
\]
Surprisingly, Zuiddam~\cite{zuiddam2018asymptotic}  proved a stronger result, using Strassen's theory of asymptotic spectra (cf.~\cite{strassen1988asymptotic,phd}):

\begin{theorem}[\cite{zuiddam2018asymptotic}] \label{thm: strassen dual}
For any graph $G$ we have
\[
\Theta(G)=\min\left\{\phi(G):~\phi\in\bX(\cG,\leq_{loc})\right\},~\overline{\chi}_f(G)=\max\left\{\phi(G):~\phi\in\bX(\cG,\leq_{loc})\right\}.
\]
% Moreover, if there exists $\phi\in\bX(\cG,\leq_{loc})$ such that $\phi(G)\geq 1$, then
% \[
% \overline{\chi}_f(G)=\max\left\{\phi(G):~\phi\in\bX(\cG,\leq_{loc})\right\}.
% \]
\end{theorem}
A follow-up work~\cite{Li2018quantum} extends the above theorem to the quantum setting. The quantum independence number $\alpha_q(G)$~\cite{manvcinska2016quantum} and the quantum clique cover number $\overline{\chi}_q(G)$~\cite{Cameron07quantumchrom} are quantum analogues of $\alpha(G)$ and $\overline \chi(G)$, defined as
\[
\alpha_q(G)=\max\{d:~\overline{K}_d\leq_{q} G\} \quad \text{ and } \quad \overline{\chi}_q(G)=\min\{d:~G\leq_{q} \overline{K}_d\}.
\]
 Then the {\em quantum Shannon capacity} is defined as
\[
\Theta_q(G):=\sup_{k}\sqrt[k]{\alpha_q(G^{\boxtimes k})}=\lim_{k\to\infty}\sqrt[k]{\alpha_q(G^{\boxtimes k})}
\]
and the asymptotic {\em quantum clique cover number} is defined as
\[
\overline{\chi}_{q,\infty}(G)=\inf_{k}\sqrt[k]{\overline{\chi}_q(G^{\boxtimes k})}=\lim_{k\to\infty}\sqrt[k]{\overline{\chi}_q(G^{\boxtimes k})}.
\]
The following quantum analogue of \cref{thm: strassen dual} holds.

\begin{theorem}[\cite{Li2018quantum}] \label{thm: quantum strassen dual}
For any graph $G$ we have
\[
\Theta_q(G)=\min\left\{\phi(G):~\phi\in\bX(\cG,\leq_{q})\right\},~\overline{\chi}_{q,\infty}(G)=\max\left\{\phi(G):~\phi\in\bX(\cG,\leq_{q})\right\}.
\]
% if there exists $\phi\in\bX(\cG,\leq_{q})$ such that $\phi(G)\geq 1$, then
% \[
% \overline{\chi}_{q,\infty}(G)=\max\left\{\phi(G):~\phi\in\bX(\cG,\leq_{q})\right\}.
% \]
\end{theorem}
We point out that \cref{thm: strassen dual,thm: quantum strassen dual} follow from an abstract theorem for preordered semirings (cf.~\cite[Theorem 2.13, Corollary 2.14-2.15]{phd}). For our purpose, it is sufficient to collect that $\leq_{loc}$ and $\leq_q$ satisfy the following properties:
\begin{proposition}\label{prop: strassen preorder for loc and q}
Let $t\in\{loc,q\}$. The preorder $\leq_t$ satisfies the following properties:
\begin{enumerate}[label=\upshape(\roman*)]
\item $d\leq d'$ in $\N$ \ $\Longleftrightarrow$\  $\overline{K}_d\leq_t\overline{K}_{d'}$ in $\cG$.
\item If $G\leq_t H$ and $K\leq_t L$, then $G\sqcup K\leq_t H\sqcup L$ and $G\boxtimes K\leq_t H\boxtimes L$.
\item For every $G,H\in\cG$ with $H\neq K_0$, there exists a finite $d \in \N$, such that $G \leq_t \overline{K_d} \boxtimes H$.
\end{enumerate}
\end{proposition}
\begin{remark}{\rm
These quantum analogues of $\alpha(G)$ and $\overline \chi(G)$ are defined by allowing (finite-dimensional) quantum strategies in corresponding nonlocal games. Another possible quantization is to consider entanglement-assisted zero-error information transmission, which is also studied in~\cite{Li2018quantum}. A major open problem in zero-error quantum information theory is whether these two quantum notions are identical. In this paper we only focus on the nonlocal game setting.}
\end{remark}

\subsection{Notation and preliminaries about operator algebras}\label{subsec: operator algebra}
We present necessary preliminary knowledge about operator algebra. See~\cite{blackadar2006operator,kadison1986fundamentalsI} for a more detailed introduction. Let $\cH$ be a Hilbert space with inner product $\langle\cdot,\cdot\rangle_\cH$. We denote the norm on $\cH$ as $\norm{\cdot}_\cH:=\sqrt{\langle\cdot,\cdot\rangle_\cH}$ or simply $\norm{\cdot}$ if there is no ambiguity. Let $B(\cH)$ be the space of bounded operators from $\cH$ to $\cH$, equipped with the operator norm $\norm{X}_{B(\cH)}=\sup\{\norm{X\eta}_\cH:~\eta\in\cH,~\norm{\eta} \leq 1\}$. We denote the identity operator in $B(\cH)$ by $I_{\cH}$. We omit the subscript if it is clear in the context.
For $X\in B(\cH)$, we write $\ran(X)=\{X\eta:~\eta\in\cH\}$ (resp. $\cl(\ran(X))$) denote the range of $X$ (resp. the closure of $\ran(X)$) and $\ker(X)=\{\eta\in\cH:~X\eta=0\}$ denotes the kernel of $X$.
 The right support projection $P_X$ of $X$ is the projection onto $\ker(X)^\perp$ and the left support projection $Q_X$ of $X$ is the projection onto $\cl(\ran(X))$. The right and left support projections satisfy the following relations:
\begin{equation}\label{eq: left and right support projection}
\begin{split}
XP_X=X~\text{and}~P_XY=0~\text{whenever}~XY=0,\\
Q_XX=X~\text{and}~YQ_X=0~\text{whenever}~YX=0.
\end{split}
\end{equation}
A sequence of operators $\{T_n\}_{n\in \N}$ converges to $T$ in the operator norm topology if $\norm{T_n-T}_{B(\cH)} \to 0$. We denote this using $\displaystyle \lim_{n\to\infty} T_n=T$. A sequence of operators $\{T_n\}_{n\in \N}$ converges to $T$ in strong operator topology (in short, SOT) if for any $\eta\in \cH$, we have $\norm{T_n\eta - T\eta}_{\cH}\to 0$. We write $\displaystyle \sotlim_{n\to\infty} T_n=T$.  In general, the strong operator topology is weaker than the operator norm topology in $B(\cH)$, but they are the same if $\cH$ is finite dimensional.

% We list some basics of $C^*$-algebras and von Neumann algebras.
A $C^*$-algebra $\cA$ is a norm-closed $*$-subalgebra (closed under adjoint) of $B(\cH)$ for some Hilbert space $\cH$. We denote the identity element of $\cA$ by $I_{\cA}$ (or $I$ if there is no confusion). A {\em state} $\tau$ on $\cA$ is a linear functional $\tau:\cA\to\C$ that is {\em positive}, i.e., $\phi(X^*X)\geq 0$ for any $X\in\cA$, and {\em unital}, i.e.,  $\phi(I_\cA)=1$.
For example, a unit vector $\eta\in \cH$ with $\norm{\eta}_{\cH}=1$ gives a state $\tau_\eta(X)=\langle\eta, X \eta\rangle$. We call such states vector states. We say a state $\tau:\cA\to \mathbb{C}$ is tracial if $\tau(XY)=\tau(YX)$ for all $X,Y\in \cA$. The normalized trace $\tr: M(d,\C)\to\C$ is a tracial state on the matrix algebra $M(d,\C)$. We specify $\Tr(X)=\sum_{i=1}^k X_{i,i}$ as the (unnormalized) trace of matrix $X$. 

A von Neumann algebra $\cM$ is a SOT-closed $*$-subalgebra of $B(\cH)$ for some $\cH$.
Since the strong operator topology is weaker than the norm topology in $B(\cH)$, every von Neumann algebra is a (unital) $C^*$-algebra. All $C^*$-algebras mentioned in this paper are unital (i.e., contain an identity element) and all mentioned von Neumann algebras are finite (i.e., admit a trace).

In contrast to $C^*$-algebras, a von Neumann algebra contains ``enough'' projections. That is, if $X\in\cM$, then both its left support projection $Q_X$ and right support projection $P_X$ are also elements of $\cM$. Moreover, a self-adjoint element $T=T^*\in \cM$  admits a spectral decomposition in $\cM$:
\[T=\int_{-\infty}^\infty\lambda d E_{\lambda}\]
where $E_{\lambda}$ is the spectral projection of $T$ onto the interval $(-\infty,\lambda]$. For a bounded Borel measurable function $f:\mathbb{R}\to \mathbb{R}$, the operator $f(T)$ by spectrum theorem is given by
\[  f(T)=\int_{-\infty}^\infty f(\lambda) d E_{\lambda}\ .\]
If $T\in B(\cH)$ is normal (i.e., $TT^*=T^*T$), its left support projection $Q_X$ and right support projection $P_X$ coincide. If furthermore $T\geq 0$, then $T(T+\eps I)^{-1}\to P_T$ as $\eps\to 0$ in the strong operator topology.

Let $P,Q\in B(\cH)$ be two projections. We denote $P\vee Q$ as the (union) projection onto the closure of $\ran(P)+\ran(Q)$ and $P\wedge Q$ as the intersection projection onto $\ran(P)\cap\ran(Q)$. We have by functional calculus that (cf.~\cite[I.5.2.1]{blackadar2006operator})
\[ \text{sot-}\lim_{n\to \infty} (P+Q)^\frac{1}{n}=P\vee Q \ ,\ \text{sot-}\lim_{n\to \infty} (PQP)^n=P\wedge Q\ .\]
If $\ran(P)\subseteq\ran(Q)$, then $PQ=QP=P$ and $P\leq Q$ in the sense that $\langle\eta,P\eta\rangle\leq\langle\eta,Q\eta\rangle$ for every $\eta\in\cH$. Moreover, by SOT-closedness, if $P$ and $Q$ belong to some von Neumann algebra $\cM$, then $P\wedge Q$ and $P\vee Q$ also belong to $\cM$.

A state $\tau:\cM\to\C$ is normal if it is weak$^*$-continuous, equivalently,
for any monotone increasing net of positive operators $\{X_\alpha\}\subseteq \cM$ with $X_\alpha \to X$, we have $\tau(X_\alpha)\to\tau(X)$. Vector states are always normal and normal states preserves the SOT-limit. As an example, for two projections $P,Q\in\cM$ and a normal state $\displaystyle\tau:\cM\to\C$ we have $\tau(P\wedge Q)=\lim_{n\to\infty}\tau((PQP)^n)$.

Let $\cA_1\subseteq B(\cH_1)$  and $\cA_2\subseteq B(\cH_2)$ be two $C^*$-algebras. Their direct sum $\cA_1\oplus \cA_2$ is again a $C^*$-algebra and similarly for two von Neumann algebras $\cM_1\subseteq B(\cH_1)$  and $\cM_2\subseteq B(\cH_2)$.
% It is straightforward to verify that the direct sum $\tau_1\oplus\tau_2$ of two (normal) tracial states $\tau_1$ on $\cA_1$ (resp. $\cM_1$) and $\tau_2$ on $\cA_2$ (resp. $\cM_2$) is a (normal) tracial state on $\cA_1\oplus\cA_2$ (resp. $\cM_1\oplus\cM_2$).
A $*$-homomorphism from $\cA_1$ to $\cA_2$ is a linear map $\pi:\cA_1\to \cA_2$ which preserves multiplication and the $*$-operation. We say $\cA_1$ and $\cA_2$ are $*$-isomorphic if there is a one-to-one and onto $*$-homomorphism from $\cA_1$ to $\cA_2$.  Let $\cA_1\odot\cA_2$ denote the the algebraic tensor product of $\cA_1$ and $\cA_2$.
The minimal $C^*$-algebraic tensor product $\cA_1\otimes_{\min} \cA_2$ is the norm closure of $\cA_1\odot \cA_2 \subseteq B(\cH_1\otimes \cH_2)$. It does not depend on the choice of embeddings $\cA_1\subseteq B(\cH_1)$ and $\cA_2\subseteq B(\cH_2)$. Let $\phi_1$ and $\phi_2$ be states on $\cA_1$ and $\cA_2$, respectively. The tensor product state $\phi_1\otimes\phi_2:\cA_1\odot \cA_2\to \mathbb{C}$ is defined by $(\phi_1\otimes \phi_2)(X_1\otimes X_2)=\phi_1(X_1)\phi_2(X_2)$ for $X_1\in\cA_1$, $X_2\in \cA_2$. It extends to a state on $\cA_1 \otimes_{\min} \cA_2$ (and also to the maximal tensor product $\cA_1 \otimes_{\max} \cA_2$, see \cite[II.9.3.5]{blackadar2006operator}). This state extension is the only property needed of the $C^*$-algebraic tensor product that we use in this paper and we therefore simply write $\cA_1 \otimes \cA_2$ for $\cA_1 \otimes_{\min} \cA_2$.

For two von Neumann algebras $\cM_1\subseteq B(\cH_1)$ and $\cM_2 \subseteq B(\cH_2)$, the von Neumann algebra tensor product is $\cM_1\overline{\otimes} \cM_2$ is the closure of $\cM_1 \odot \cM_2$ in the strong operator topology, i.e., $\cM_1 \overline \otimes \cM_2 =\overline{(\cM_1\odot \cM_2)}^{\mathrm{sot}} \subseteq B(\cH_1\otimes \cH_2)$. $\cM_1 \overline \otimes \cM_2$ is also independent of the choice of embeddings $\cM_1 \subseteq B(\cH_1)$ and $\cM_2 \subseteq B(\cH_2)$. Given two normal states $\phi_1$ and $\phi_2$  on $\cM_1$ and $\cM_2$ respectively, the tensor product state $\phi_1\otimes\phi_2$ is also a normal state on $\cM_1\overline{\otimes} \cM_2$. If there is no risk of confusion, we write $\cM_1\otimes \cM_2$ for $\cM_1\overline{\otimes} \cM_2$.

\subsection{An infinite-dimensional rank-nullity theorem}
Let $\cH$ be a Hilbert space and $S,T\subseteq \cH$ be two subspaces of $\cH$. When $\cH$ is finite dimensional, we have the following well-known identity
\begin{align}
\dim(S)+\dim(T)=\dim(S+T)+\dim(S\cap T). \label{eq:rank}
\end{align}
To see the above as a statement about ranks of operators, let $P$ be the orthogonal projector on $S$, and $Q$ the orthogonal projector on $T$. Then the above identity becomes
\[
\rk(P) + \rk(Q) = \rk(P\vee Q) + \rk(P\wedge Q).
\]
Moreover, if $S$ and $T$ are the row spaces of $A$ and $B$ in $B(\cH)$, respectively, we can compute the rank of $AB^T$ by
\[
\rk(AB^T)=\dim(S)-\dim(S\cap T^\perp)=\rk(P)-\rk(P\wedge(I-Q)).
\]
In the infinite-dimensional setting, the dimension of closed subspaces can be interpreted as the trace of their corresponding projections, a similar identity can be obtained as a consequence from Kaplansky's formula of projections (c.f.~\cite[Theorem 6.1.7]{kadison1986fundamentalsII}). This identity (and related results) will be a key tool to our proofs.  

\begin{restatable}{lemma}{infrank}\label{thm:rank} Let $\cM\subseteq B(\cH)$ be a von Neumann algebra and let $P,Q\in
\cM$ be two projections.
Suppose $\tau:\cM\to \mathbb{C}$ is a normal tracial state. Then
\begin{enumerate}[label=\upshape(\roman*)]
\item $\tau(P)+\tau(Q)=\tau(P\vee Q)+\tau(P\wedge Q)$.
\item Let $Y\in\cM$ be the projection onto $\cl(\ran(QP))$. Then we have
\[
\tau(Q)\geq\tau(Y)=\tau(P)-\tau(P\wedge (I_\cM-Q)).\]
\item If additionally $P\wedge (I_\cM-Q)=0$, then $\tau(Y)=\tau(P)$. Moreover, for any $\epsilon>0$, there exists projection $P_\eps \in\cM$ such that $P_\eps\le P$, $\tau(P-P_\eps)\le \epsilon$, and $\norm{P_\eps(I_\cM-Q)}<1$.
\end{enumerate}
\end{restatable}
\begin{proof}
Recall the two projections $P,Q\in \cM$ are Murray-von Neumann equivalent, denoted by $P\sim Q$, if there exists a partial isometry $V\in\cM$ such that $V^*V=P$ and $VV^*=Q$. It follows from the property of tracial states that $\tau(P)=\tau(Q)$ if $P\sim Q$. Then (i) follows from Kaplansky's formula of projections which states that for any two projections $P,Q\in \cM$
\[P\vee Q-P \sim Q-P\wedge Q\ ,\]
% where $P\vee Q-P$ is the left support projection of $(I-P)Q$ and $Q-P\wedge Q$ is the left support projection of $Q(I-P)$.
where the ``$-$'' operation is interpreted in the projection lattice of equivalence classes under $\sim$ (see \cite[Chapter 2]{kadison1986fundamentalsI}). 
For (ii), it is clear that $\ran(Y)\subseteq \ran(Q)$, hence $Y\le Q$ and $\tau(Y)\le \tau(Q)$. For the equality, let $X$ denote the projector onto $\cl(\ran(PQ))$. It follows from~\cite[Proposition 6.1.6]{kadison1986fundamentalsII} that $X\sim Y$. Moreover, $\ran(X)^\perp=\cl(\ran(PQ))^\perp=\ker(QP)=\ker{P}\oplus (\ran(P) \cap \ker(Q) )$. Hence
\[1-X\sim (I_\cM-P)+P\wedge(I_\cM-Q)\ , \quad Y\sim X\sim P-P\wedge(I-Q)\ ,\]
which proves (ii) by taking trace. 

We give a detailed proof for (iii). Note that if $P\wedge (I_\cM-Q)=0$, then $\text{sot-}\lim_{n\to \infty}(P(I_\cM-Q)P)^{n}\to 0$ (monotone decreasingly) and $\tau(Y)=\tau(P)-\tau(P\wedge (I_\cM-Q))=\tau(P)$. 
For any $\epsilon>0$, there exists $n_0 \in \N$ such that
$ \tau((P(I_\cM-Q)P)^{n_0})\le \epsilon/2 $. Choose $\delta>0$ such that $(1-\delta)^{n_0}\ge \frac{1}{2}$.
Let $P_\delta$ be the spectral projection of $P(I_\cM-Q)P$ corresponding to the interval $(1-\delta, 1]$. Thus $P_\delta\leq P$ and $(1-\delta)P_\delta\leq P(I_\cM-Q)P$. This implies
\[ \frac{1}{2}\epsilon\ge \tau((P(I_\cM-Q)P)^{n_0})\ge \tau((1-\delta)^{n_0}P_\delta)= (1-\delta)^{n_0} \tau(P_\delta) .\]
Because $(1-\delta)^{n_0}\ge 1/2$, we have
\[ \tau(P_\delta)\le \frac{\tau((P(I_\cM-Q)P)^{n_0})}{(1-\delta)^{n_0}}\le \epsilon\ .\]
Note that $\ran(P_\delta)\subseteq \ran(P(I_\cM-Q)P)\subseteq \ran(P)$. Then $P_\delta\le P$ and we define the projection $P_\eps=P-P_\delta$. We have $\tau(P-P_\eps)=\tau(P_\delta)\le \epsilon$ and by the definition of $P_\delta$,
\[
\norm{P_\eps(I_\cM-Q)}^2  = \norm{P_\eps(I_\cM-Q)P_\eps}=\norm{(I_\cM-P_\delta)P(I_\cM-Q)P(I_\cM-P_\delta)}\le 1-\delta,\]
where the equality uses the fact that $P_\delta P=P_\delta$ and the inequality uses the fact that $I_\cM-P_\delta$ is the spectral projection of $P(I_\cM-Q)P$ corresponding to the interval $[0,1-\delta]$.
\end{proof}

We remark that the above lemma can be alternatively proved via the representation theorem \cite[Theorem~1.3]{twoprojections} of the universal $C^*$-algebra generated by two projections. 
% This algebra is $*$-isomorphic to
% %\begin{align}
% $\cA:=\{ f\in C([0,1],M_2(\C))\, \colon f(0),f(1) \text{ diagonal} \}$ 
% %\label{eq:universal} \end{align}
% where $C([0,1],M_2(\C))$ is the set of $2\times 2$ matrix-valued continuous functions and the images of $P$ and $Q$ in $\cA$ are
% %\begin{align}
% \[
% P(t)=\begin{bmatrix}1& 0\\0&0\end{bmatrix}, \quad
% Q(t)=\begin{bmatrix}t& \sqrt{t(1-t)}\\\sqrt{t(1-t)}&1-t\end{bmatrix} \quad \text{for}~t\in[0,1],
% \]
% %\end{align}
% respectively. 
This representation theorem is also used in the study of the tracial rank \cite{synchronous}.

\section{ \texorpdfstring{$C^*$}{C*}-algebraic and commuting quantum Shannon capacity and their asymptotic spectral characterizations}\label{sec: asymptotic spectra}

\subsection{ \texorpdfstring{$C^*$}{C*}-algebraic and commuting quantum homomorphisms}
We first recall the definitions of graph homomorphism via $C^*$-algebras that were considered in~\cite{ORTIZ2016}. They provide infinite-dimensional analogues of the previously mentioned notions of (quantum) graph homomorphism.

\begin{definition}\label{def: C*homomorphism}
For graphs $G$ and $H$, we say $G\Cto H$ if there exists a $C^*$-algebra $\cA(G,H)$ containing projections $E_{g,h}\in \cA(G,H)$ for every $g\in V(G)$ and $h\in V(H)$, satisfying the conditions:
\begin{enumerate}[label=\upshape(\roman*)]
\item $\sum_{h\in V(H)} E_{g,h}=I$ for any $g\in V(G)$,
% \item $E_{g,h}E_{g,h'}=0$ for any $g\in V(G)$ and  $h\neq h'\in V(H)$,
\item $E_{g,h}E_{g',h'}=0$ if $\{g,g'\}\in E(G)$ and ($\{h,h'\}\in E(\overline{H})$ or $h=h'\in V(H)$).
\end{enumerate}
\end{definition}
\begin{definition}\label{def: qc homomorphism}
For graphs $G$ and $H$, we say $G\QCto H$ if there exists a $C^*$-algebra $\cA(G,H)$ equipped with a tracial state $\tau$ and  containing projections $E_{g,h}\in\cA(G,H)$ for all $g\in V(G)$ and $h\in V(H)$, satisfying the conditions:
\begin{enumerate}[label=\upshape(\roman*)]
\item $\sum_{h\in V(H)} E_{g,h}=I$ for any $g\in V(G)$,
% \item $E_{g,h}E_{g,h'}=0$ for any $g\in V(G)$ and  $h\neq h'\in V(H)$,
\item $E_{g,h}E_{g',h'}=0$ if $\{g,g'\}\in E(G)$ and ($\{h,h'\}\in E(\overline{H})$ or $h=h'\in V(H)$).
\end{enumerate}
\end{definition}

Note that \cref{def: C*homomorphism} reduces to \cref{def: qc homomorphism} if $\cA(G,H)$ admits a tracial state. \cref{def: C*homomorphism,def: qc homomorphism} reduce to the usual notion of graph homomorphism if $\cA(G,H)$ has a $1$-dimensional representation, and they reduce to the quantum homomorphism if $\cA(G,H)$ has a finite-dimensional representation.
Hence we have the following chain of implications:
\begin{equation}\label{eq: preorder relation}
G\to H~\Longrightarrow~G\Qto H~\Longrightarrow~G\QCto H~\Longrightarrow~G\Cto H  ~\Longrightarrow~G\Bto H \Longleftrightarrow \vartheta (G)\le \vartheta(H),
\end{equation}
where the right most implication is shown in~\cite{ORTIZ2016} and the last equivalence follows from \cref{prop: preorder definition for lovasz theta}.

These graph homomorphisms permit to define the corresponding $C^*$-algebraic and commuting quantum analogues of the independence number, clique cover number and asymptotic spectrum.
\begin{definition}
Let $t\in\{qc,C^*\}$ and $G$ and $H$ be graphs. Define
\begin{itemize}
\item $G\leq_t H$ if $\overline{G}\tto\overline{H}$;
\item $\alpha_t(G):=\max\{d:~K_d\tto\overline{G}\}=\max\{d:~\overline{K_d}\leq_t G\}$, and $\displaystyle \Theta_t(G):=\lim_{k\to\infty}\sqrt[k]{\alpha_t(G^{\boxtimes k})}$;
\item $\overline{\chi}_t(G):=\min\{d:~\overline{G}\tto K_d\}=\min\{d:~G\leq_t\overline{K_d}\}$, and $\displaystyle\overline{\chi}_{t,\infty}(G):=\lim_{k\to\infty}\sqrt[k]{\overline{\chi}_t(G^{\boxtimes k})}$;
\item $\bX(\cG,\leq_t):=\{\phi\in\Hom(\cG,\R_{\geq 0}):~\forall\ G,H\in \cG,~G\leq_t H\Rightarrow \phi(G)\leq \phi(H)\}$.
\end{itemize}
\end{definition}
By Fekete's lemma~\cite{Fekete1923}, both $\Theta_t(G)$ and $\overline{\chi}_{t,\infty}(G)$ are well-defined and equal $\sup_{k\in\N}\sqrt[k]{\alpha_t(G^{\boxtimes k})}$ and $\inf_{k\in\N}\sqrt[k]{\overline{\chi}_t(G^{\boxtimes k})}$, respectively.
By \cref{eq: preorder relation}, we obtain the chain of inclusions:
\begin{equation}\label{eq: Asymptotic spectrum relation}
\bX(\cG,\leq_{C^*})\subseteq\bX(\cG,\leq_{qc})\subseteq\bX(\cG,\leq_q)\subseteq\bX(\cG,\leq_{loc}).
\end{equation}

Our first result shows that $\leq_{qc}$ and $\leq_{C^*}$ define two new preorders on $\cG$ and satisfy similar properties as those  in \cref{prop: strassen preorder for loc and q} for $t\in \{loc,q\}$.
\begin{proposition}\label{prop: strassen preorder}
For graphs $G,H,K,L\in\cG$ and $t\in\{qc,C^*\}$, we have
\begin{enumerate}[label=\upshape(\roman*)]
\item $G\leq_t G$,
\item If $G\leq_t H$ and $H\leq_t K$, then $G\leq_t K$,
\item $\overline{K_d} \leq_t \overline{K_{d'}}$ if and only if $d \leq d'$ in $\N$,
\item If $G \leq_t H$ and $K \leq_t L$, then $G\sqcup K \leq_t H\sqcup L$ and $G \boxtimes K \leq_t H\boxtimes L$,
\item If $H\neq K_0$, then there exists $d \in \N$ with $G \leq_t \overline{K_d} \boxtimes H$.
\end{enumerate}
\end{proposition}
\begin{proof}
(i) is clear.

(ii) Let $\cA(G,H)$ be a $C^*$-algebra containing a set of projections $\{E_{g,h}:~g\in V(G),h\in V(H)\}\subseteq \cA(G,H)$, and let $\cA(H,K)$ be a $C^*$-algebra containing a set of projections $\{F_{h,k}:~h\in V(H),k\in V(K)\}\subseteq \cA(H,K)$, which are feasible solutions of $G\leq_{C^*} H$ and $H\leq_{C^*} K$, respectively.
Define  $A_{g,k}=\sum_{h\in V(H)}E_{g,h}\otimes F_{h,k}$ for $g\in V(G)$ and $k\in V(K)$ and let $\cA(G,K)\subseteq\cA(G,H)\otimes\cA(H,K)$ be the $C^*$-algebra generated by $\{A_{g,k}:~g\in V(G),~k\in V(K)\}$. In~\cite[Prop.~4.6]{ORTIZ2016} it is shown that $\cA(G,K)$ is a feasible solution of $G\leq_{C^*} K$. For $t=qc$, let in addition  $\tau_{G,H}$ be a tracial state on $\cA(G,H)$ and $\tau_{H,K}$ be a tracial state on $\cA(H,K)$, respectively. Then $\tau_{G,K}=\tau_{G,H}\otimes\tau_{H,K}$ is a tracial state on $\cA(G,K)$, since $\tau_{G,K}$ is a tracial state on $\cA(G,H)\otimes\cA(H,K)$~\cite[II.9.3.5]{blackadar2006operator}.

(iii)
By \cref{eq: preorder relation}, $\overline {K_d} \le_t \overline {K_{d'}}$ implies $\vartheta(\overline {K_d} ) \le \vartheta(  \overline{K_{d'}})$ and thus $d\le d'$.
Conversely, $d\le d'$ implies $\overline {K_d}\le \overline {K_{d'}}$ by \cref{prop: strassen preorder for loc and q}, which, using again \cref{eq: preorder relation}, implies $\overline{K_d}\le_t \overline{K_{d'}}$.

(iv) Let $\cA(G,H)$ be a $C^*$-algebra containing projections $\{E_{g,h}:~g\in V(G),h\in V(H)\}$ and $\cA(K,L)$ be a $C^*$-algebra containing  projections $\{F_{k,\ell}:~k\in V(K),\ell\in V(L)\}$, which are feasible solutions of $G\leq_{C^*} H$ and $K\leq_{C^*} L$, respectively.

First we show $G\sqcup K \leq_{C^*} H\sqcup L$. For this,  let $\cA(G\sqcup K, H\sqcup L)\subseteq \cA(G,H)\otimes\cA(H,K)$ be the $C^*$-algebra generated by $\{A_{u,v}:~u\in V(G\sqcup K),~v\in V(H\sqcup L)\}$, where
\[
A_{u,v} = \begin{cases} E_{u,v}\otimes I_{\cA(K,L)} & \textnormal{if $u \in V(G), v\in V(H)$},\\
I_{\cA(G,H)}\otimes F_{u,v} & \textnormal{if $u \in V(K), v\in V(L)$},\\
0 & \textnormal{otherwise.}
\end{cases}
\]
It is straightforward that $A_{u,v}$ is a projection for $u\in V(G\sqcup K)$ and $v\in V(H\sqcup L)$. For a fixed $u\in V(G\sqcup K)$, if $u\in V(G)$,
\[
\sum_{v\in V(H \sqcup L)}A_{u,v}=\sum_{v\in V(H)}E_{u,v}\otimes I_{\cA(K,L)}=I_{\cA(G,H)}\otimes I_{\cA(K,L)}.
\]
Similarly, if $u\in V(K)$,
\[
\sum_{v\in V(H\sqcup L)}A_{u,v}=\sum_{v\in V(L)}I_{\cA(G,H)}\otimes F_{u,v}=I_{\cA(G,H)}\otimes I_{\cA(K,L)}.
\]
Lastly, we show that if $\{u,u'\}\in E(\overline{G\sqcup K})$ and ($\{v,v'\}\in E(H\sqcup L)$ or $v=v'$), then $A_{u,v}A_{u',v'}=0$ holds. Note that $\{v,v'\}\in E(H\sqcup L)$ or $v=v'$ implies that $v,v'\in V(H)$ or $v,v'\in V(L)$ (since there is no edge between $V(H)$ and $V(L)$). We distinguish the following cases:
\begin{itemize}
\item Either,  $u,u'\in V(G)$ and $v,v'\in V(H)$, then $\{u,u'\}\in E(\overline G)$ and thus $E_{u,v}E_{u',v'}=0$, implying $A_{u,v}A_{u',v'}=E_{u,v}E_{u',v'}\otimes I_{\cA(K,L)}=0$.
\item Or, $u,u'\in V(K)$ and $v,v'\in V(L)$, then $\{u,u'\}\in E(\overline K)$ and thus $F_{u,v}F_{u',v'}=0$, implying $A_{u,v}A_{u',v'}=I_{\cA(G,H)}\otimes F_{u,v}F_{u',v'}=0$.
\item In all other cases, we have $A_{u,v}=0$ or $A_{u',v'}=0$, thus $A_{u,v}a_{u',v'}=0$.
\end{itemize}
Thus we have shown $G\sqcup K \le_{C^*} H\sqcup L$.

To show $G\sqcup K\leq_{qc} H\sqcup L$ one can use the above construction to build a $C^*$-algebra $\cA(G\sqcup K,H\sqcup L)$ and the tracial state can be $\tau(G,H)\otimes\tau(K,L)$ as in (ii).

We now turn to show $G \boxtimes K \leq_{C^*} H\boxtimes L$. For this,   let $\cA(G\boxtimes K, H\boxtimes L)\subseteq \cA(G,H)\otimes\cA(K,L)$ be the $C^*$-algebra generated by $\{A_{(g,k),(h,\ell)} := E_{g,h} \otimes F_{k,\ell}:~g\in V(G),\ h\in V(H),\ k\in V(K),\ l\in V(L)\}$. Then it is easy to verify that $\cA(G\boxtimes K, H\boxtimes L)$ satisfies \cref{def: C*homomorphism}. For $G \boxtimes K \leq_{qc} H\boxtimes L$, we use again the product tracial state $\tau(G,H)\otimes\tau(K,L)$ on $\cA(G,H)\otimes\cA(K,L)$.

(v) If $H\neq \overline{K_0}$, then there is a $d \in \N$ with $G\leq_{loc} \overline{K_d} \boxtimes H$ (see \cite{zuiddam2018asymptotic}). Then the result follows using \cref{eq: preorder relation}.
\end{proof}
Using the above proposition, we can again invoke the abstract theorems for preordered semirings (cf.~\cite[Theorem 2.13, Corollary 2.14-2.15]{phd}), that gave us \cref{thm: strassen dual} and \cref{thm: quantum strassen dual}. We obtain the following analogue for the cases $t \in \{qc,C^*\}$.
\begin{theorem}\label{thm: dual characterization}
Let $t\in\{qc,C^*\}$ and $G\in\cG$. We have
\begin{equation}\label{eq: dual of shannon}
\Theta_t(G)=\min\{\phi(G):~\phi\in\bX(\cG,\leq_t)\},~\overline{\chi}_{t,\infty}(G)=\max\{\phi(G):~\phi\in\bX(\cG,\leq_t)\}.
\end{equation}
% If there exists $\phi\in\bX(\cG,\leq_t)$ such that $\phi(G)\geq 1$, then
% \begin{equation}\label{eq: dual of asymptotic clique cover}
% \overline{\chi}_{t,\infty}(G)=\max\{\phi(G):~\phi\in\bX(\cG,\leq_t)\}.
% \end{equation}
\end{theorem}

\subsection{Elements in the new quantum asymptotic spectra of graphs}\label{sec: spectral points}

It is proved in~\cite{ORTIZ2016} that $G\leq_{C^*} H$ implies $G\leq_B H$, where the latter is equivalent to $\vartheta(G)\leq\vartheta(H)$. Thus, $\vartheta$ is a monotone for both preorders $\leq_{C^*}$ and $\leq_{qc}$. Therefore,  \[\vartheta\in\bX(\cG,\leq_{C^*})\subseteq \bX(\cG,\leq_{qc}).\]
We prove that the tracial rank introduced in~\cite{PAULSEN2016} is an element of $\bX(\cG,\leq_{qc})$. In view of the inclusions in \cref{eq: Asymptotic spectrum relation}, we also obtain a new element in $\bX(\cG,\leq_q)$ and $\bX(\cG,\leq_{loc})$. In~\cite{PAULSEN2016}, the multiplicativity (w.r.t.~strong product) and monotonicity (w.r.t.~commuting quantum homomorphism) have been verified. We shall prove its additivity under disjoint union and show it is normalized. We first recall the definition of the tracial rank.
\begin{definition}\label{def: tracial rank}
We say that a graph $G$ has a $\lambda$-tracial representation if there exists a $C^*$-algebra $\cA$ equipped with a tracial state $\tau$ and projections $E_{g}\in \cA$ for $g\in V(G)$ which satisfy the conditions:
\begin{enumerate}[label=\upshape(\roman*)]
\item $E_gE_{g'}=0$ for all $\{g,g'\}\in E(\overline{G})$, and
\item $\tau(E_g)=\frac{1}{\lambda}$ for all $g\in V(G)$.
\end{enumerate}
The \emph{tracial rank} of $G$ is defined as
\begin{equation}
\xitr(G)=\inf\{\lambda: G~{\rm has~a}~\lambda\text{-tracial representation}\}.
\end{equation}
\end{definition}

Note that \cref{def: tracial rank} coincides with the one in~\cite{PAULSEN2016} except we replace the graph by its complement. It is shown in~\cite[Theorem 6.11]{PAULSEN2016} that for every feasible solution $(\cA,\{E_g\},\tau)$, its GNS construction is also a feasible solution of $\overline{\xi}_{tr}$, and the second condition can be relaxed to
\begin{itemize}
\item[(ii)] $\tau(E_g)\geq\frac{1}{\lambda}$ for all $g\in V(G)$.
\end{itemize}
We use the above two formulations to show that the tracial rank is additive with respect to taking disjoint union.

\begin{proposition}\label{prop: additivity}
For graphs $G,H\in\cG$, we have $\xitr(G\sqcup H)=\xitr(G)+\xitr(H)$.
\end{proposition}

\begin{proof}
We first prove $\xitr(G\sqcup H)\leq \xitr(G)+\xitr(H)$. Let $(\cA_G,\{E_g\}_{g\in V(G)},\tau_G)$ be a $\lambda_G$-tracial representation of $\xitr(G)$ and let $(\cA_H,\{F_h\}_{h\in V(H)},\tau_H)$ be a $\lambda_H$-tracial representation of $H$, as given in~\cref{def: tracial rank}.
Let $\cA=\cA_G\oplus\cA_H$ be the direct sum of $\cA_G$ and $\cA_H$, which is a $C^*$-algebra. $\cA$ contains projections $(E_g,0)$ and $(0,F_h)$ for any $g\in V(G)$ and $h\in V(H)$. Define the linear map \[\tau: \cA\mapsto \mathbb{C}\ , \ \tau\big( (X,Y) \big)=\frac{\lambda_G}{\lambda_G+\lambda_H}\tau_G(X)+\frac{\lambda_H}{\lambda_G+\lambda_H}\tau_H(Y)\ .\]
$\tau$ is clearly a tracial state on $\cA$.
Let $A_v\in\cA$ for $v\in V(G)\cup V(H)$ by
\[A_{v}=\begin{cases} (E_v, 0)\in  \cA_G\oplus\cA_H & \textnormal{if $v \in V(G)$},\\
 (0, F_v) \in  \cA_G\oplus\cA_H & \textnormal{if $v \in V(H)$.}\\
\end{cases}.\]
Then each $A_v$ is a projection and $A_{v}A_{u}=0$ if $\{v,u\}\in E(\overline{G\sqcup H})$. Moreover, we have
\[
\tau(A_v)=\begin{cases} \tau(E_v, 0)= \frac{\lambda_G}{\lambda_G+\lambda_H}\tau_G(E_v) = \frac{1}{\lambda_G+\lambda_H}& \textnormal{if $v \in V(G)$,}\\
 \tau(0, F_v)= \frac{\lambda_H}{\lambda_G+\lambda_H}\tau_H(F_v)= \frac{1}{\lambda_G+\lambda_H} & \textnormal{if $v \in V(H)$.}\\
\end{cases}.
\]
Hence $(\cA',\{a_v\}_{v\in V(G)\cup V(H)},\tau)$ is a $(\lambda_G+\lambda_H)$-tracial representation of $G\sqcup H$, which implies $\xitr(G\sqcup H)\leq \xitr(G)+\xitr(H)$.

We now prove $\xitr(G\sqcup H)\geq \xitr(G)+\xitr(H)$. Let $(\cA,\{A_v\}_{v\in V(G)\cup V(H)},\tau)$ be a $\lambda$-tracial representation of $G\sqcup H$. Note that (i) implies that $A_gA_h=0$ for every $g\in V(G)$ and $h\in V(H)$. Let $\cM$ be the von Neumann algebra generated by $\cA$, $\tau$ extends to $\cM$. Note that the projections $E=\bigvee_{v\in V(G)} A_v$ (the projection onto the closure of the range $\sum_{g\in V(G)} A_g$) and $F=\bigvee_{h\in V(H)} A_h$ (the projection onto the closure of the range $\sum_{u\in V(H)} A_u$) are elements of $\cM$. We have $EF=0$ and $E$ and $F$ commute with $A_g$ for any $g\in V(G)$ and $A_h$ for any $h\in V(H)$, respectively. Let $\cA_G$ be the $C^*$-algebra generated by $\{A_g:g\in V(G)\}$ with $E$ being the identity element, and $\cA_H$ be the $C^*$-algebra generated by $\{A_h:h\in V(H)\}$ with $F$ being the identity element. It is straightforward to see that $A_gA_{g'}=0$ for all $\{g,g'\}\in E(\overline{G})$ and $A_hA_{h'}=0$ for all $\{h,h'\}\in E(\overline{H})$.

We construct tracial states on $\cA_G$ and $\cA_H$. Note that $\cA_G,\cA_H\subseteq\cM$, thus the tracial state $\tau:\cM\to\C$ is also a tracial linear map on $\cA_G$ and $\cA_H$. Define $\tau_G:\cA_G\to\C$ by $\tau_G(X)=\tau(X)/\tau(E)$ for any $X\in \cA_G$ and $\tau_H:\cA_H\to\C$ by $\tau_G(Y)=\tau(Y)/\tau(F)$ for any $Y\in \cA_H$. Then $\tau_G$ and $\tau_H$ are tracial states on $\cA_G$ and $\cA_H$, respectively. Note that
\[
\tau_G(A_g)=\frac{\tau(A_g)}{\tau(E)}=\frac{1}{\lambda\tau(E)}~\text{for any}~g\in V(G)~\text{and}~\tau_H(A_h)=\frac{\tau(A_h)}{\tau(f)}=\frac{1}{\lambda\tau(F)}~\text{for any}~g\in V(G).
\]
This shows that $(\cA_G,\{A_g\}_{g\in V(G)},\tau_G)$ and $(\cA_H,\{A_h\}_{h\in V(H)},\tau_H)$ are $\lambda\tau(E)$-tracial representation of $G$ and $\lambda\tau(F)$-tracial representation of $H$, respectively. Thus,
\[
\xitr(G)+\xitr(H)\leq\lambda\tau(E)+\lambda\tau(F)\leq\lambda,
\]
where the second inequality holds since $E+F\leq I_\cA$. Taking the infimum over $\lambda$ shows that $\xitr(G)+\xitr(H)\leq\xitr(G\sqcup H)$.
\end{proof}

\begin{proposition}
The tracial rank is an element of $\bX(\cG,\leq_{qc})$.
\end{proposition}
\begin{proof}
It is easy to verify that $\xitr(K_1)=1$. Then, by the additivity of $\xitr$, we have $\xitr(\overline{K}_d)=\xitr(\sqcup_{i=1}^d K_1)=d$ for all $d\in\N$. Since tracial rank is also multiplicative and monotone, it is an element of $\bX(\cG,\leq_{qc})$.
\end{proof}

\section{Tracial Haemers bound} \label{sec: tracial haemers}
It is proved that the fractional Haemers bound and the projective rank are elements of the asymptotic spectrum $\bX(\cG,\leq_q)$~\cite{Li2018quantum}, and we just proved that the tracial rank is an element of $\bX(\cG,\leq_{qc})$. Since the tracial rank can be viewed as an infinite-dimensional generalization of the projective rank~\cite{PAULSEN2016}, a natural question is whether there is an infinite-dimensional generalization of the fractional Haemers bound.

We define the tracial Haemers bound of $G$, denoted as $\cH_{tr}(G)$. This new graph parameter is a von Neumann algebraic generalization of the fractional Haemers bound and we show that it is an element of $\bX(\cG,\leq_{qc})$. In other words, we show that it is normalized, $\leq_{qc}$-monotone, additive, and multiplicative. In particular, these also imply that $\cH_{tr}(G)$ is an upper bound on $\alpha_{qc}(G)$ and $\Theta_{qc}(G)$.

\subsection{Definition and examples}
We start with recalling a coordinate-free definition of the fractional Haemers bound due to Lex Schrijver~\cite[Remark after Prop.~$7$]{bukh2018fractional}. The key part of this definition is the notion of a subspace representation of a graph.
\begin{definition}[Subspace representation] \label{def: fractional Haemers}
We say a graph $G$ has a {$(d,r)$-subspace representation} (over $\C$) if there exist subspaces $S_g\subseteq \C^d$ for all $g\in V(G)$ satisfying
\begin{enumerate}[label=\upshape(\roman*)]
\item $S_g\cap(\sum_{g'\in N_{\overline{G}}(g)}S_{g'})=\{0\}$, where the summation is over all $g'$ satisfying $\{g,g'\}\in E(\overline{G})$.
\item $\dim(S_g)=r$ for all $g\in V(G)$;
\end{enumerate}
\end{definition}
\noindent The fractional Haemers bound (over $\C$) is then defined as
\begin{equation}
\cH_f(G;\C)=\inf\{\frac{d}{r}:~G~\text{has a}~(d,r)\text{-subspace representation over }\C\}.
\end{equation}
In the infinite-dimensional setting, the trace of a projection is an analogue of the dimension of a subspace. This suggests to define the von Neumann-algebraic generalization of a subspace representation using projections. In order to do so, recall that the intersection of two subspaces and (the closure of) the sum of two subspaces corresponds to the wedge and the union of the projections, which can be expressed as strong-operator limits. Indeed, for two projections $E,F$ in a von Neumann algebra $\cM\subseteq B(\cH)$, the projection onto $\ran(E)\cap \ran (F)$ is $E\wedge F=\sotlim_{n\to \infty} (EFE)^n$ and the projection onto $\text{cl}(\ran(E)\cup \ran (F))$ is $E\vee F=\sotlim_{n\to \infty} (E+F)^\frac{1}{n}$.
In light of \cref{def: fractional Haemers}(i), it will be convenient to use the shorthand notation
\begin{equation}\label{eq: F_g}
E_{N_{\overline{G}}(g)}:= \bigvee_{g'\in N_{\overline{G}}(g)} E_{g'}=\sotlim_{n\to\infty}\Big(\sum_{g'\in N_{\overline{G}}(g)} E_{g'}\Big)^{1/n} \text{ for } g \in V(G).
\end{equation}
We then have the following von Neumann-algebraic generalization of a subspace representation.
\begin{definition}[Tracial subspace representation] \label{def: tracial subspace rep}
We say a graph $G$ has a \emph{$\lambda$-tracial subspace representation} if there exist a von Neumann algebra $\cM$ containing projections $E_g$ for all $g\in V(G)$, and a normal tracial state $\tau:\cM \to \C$,
such that \begin{enumerate}[label=\upshape(\roman*)]
\item $E_g \wedge E_{N_{\overline G}(g)} = 0$.
\item $\tau(E_g)= \frac{1}{\lambda}$ for all $g\in V(G)$.
\end{enumerate}
\end{definition}
We define the tracial Haemers bound of $G$ as the infimum over all $\lambda$ for which there exists a $\lambda$-tracial subspace representation of $G$.
\begin{definition}[Tracial Haemers bound] \label{def: tracial Haemers}
 The {tracial Haemers bound} of a graph $G$ is defined as
\begin{align*}
\cH_{tr}(G):=&\inf\{\lambda \colon G \text{ has a } \lambda\text{-tracial subspace representation}\}.
\end{align*}
\end{definition}
It is often convenient to work with a concrete representation $\cM\subset B(\cH)$, and we will also refer to a tuple $(\cH,\cM, \{E_g\},\tau)$ as a (concrete) tracial subspace representation of $G$. Note that for such a concrete von Neumann algebra, item (i) in~\cref{def: tracial subspace rep} can be equivalently formulated as
\begin{enumerate}
\item[(i)] $S_g\cap\cl(\sum_{g'\in N_{\overline{G}}(g)}S_{g'})=\{0\}$, where $S_g=\ran(E_g)$ for all $g\in V(G)$,\footnote{The difference with item (i) of \cref{def: fractional Haemers} is the \emph{closure} of $(\sum_{g'\in N_{\overline{G}}(g)}S_{g'})$. Note however that in finite-dimensional Hilbert spaces all linear subspaces are closed.}
\end{enumerate}
which gives exactly correspondence to fractional Haemers bound. The next proposition shows a concrete tracial subspace representation is always achievable by GNS representation.

\begin{proposition}
Every $\lambda$-tracial subspace representation $(\cM, \{E_g\}, \tau)$ as in \cref{def: tracial subspace rep} yields a $\lambda$-tracial subspace representation $(\cH_\pi,\pi(\cM),\pi(E_g),\phi_\tau)$ where $\pi:\cM\to B(\cH_\pi)$ is the GNS representation with respect to $\tau$ and $\phi_\tau$ is the tracial vector state induced by $\tau$.
\end{proposition}
\begin{proof}
Note that the GNS representation $\pi$ is normal because $\tau$ is a normal state. Thus, it preserves SOT limits, and hence $\pi(E_{N_{\overline G}(g)})=\bigvee_{g'\in N_{\overline{G}}(g)} \pi(E_{g'})$. Then it follows from the condition (i) from Definition \ref{def: tracial subspace rep} that $\pi(E_g)\wedge\pi(E_{N_{\overline G}(g)})=\pi(E_g \wedge E_{N_{\overline G}(g)})=0$ for every $g\in V(G)$.
Moreover the GNS representation is such that $\phi_\tau(\pi(E_g)) = \tau(E_g)$ for all $g \in G$.
\end{proof}

\begin{remark}\label{remark}
\hspace{1cm} \\
{\rm
\indent a) In contrast to the definition of tracial rank \cite{PAULSEN2016} which uses $C^*$-algebras and tracial states, we use \emph{von Neumann algebras} and \emph{normal} tracial states in the definition. One reason is that the projection $E_{N_{\overline{G}}(g)}$ for each $g\in V(G)$ is an element in the von Neumann algebra generated by $\{E_g\}$, but not necessarily in the $C^*$-algebra generated by $\{E_g\}$. Moreover, using a normal tracial state $\tau$ ensures the following continuity: $\tau(P) = \lim_{k \to \infty} \tau(P_k)$ whenever $P = \sotlim_{k \to \infty} P_k$.  Note that the definition of the tracial rank uses only multiplication of operators and this is well-defined within in $C^*$-algebra (and therefore does not require normal states). \\
%only orthogonality relations which is well-defined within a $C^*$-algebra structure since the product of two elements still belongs to the $C^*$-algebra (and therefore does not require normal states). \\
\indent b) In the definition of a tracial subspace representation we may allow the equality in item {\upshape (ii)} to be an inequality: $\tau(E_g) \geq \frac{1}{\lambda}$ for all $g \in V(G)$. The proof is standard, we have included it in \cref{App: inequality} for completeness.
}
\end{remark}

We first show that the tracial Haemers bound is a proper infinite dimensional generalization of the fractional Haemers bound. It is clear that if a graph admits a $(d,r)$-subspace representation, then there is also a $d/r$-tracial subspace representation. The next proposition shows the converse also holds.
\begin{proposition}\label{prop: finite dim C^* algebra}
If a graph $G$ has a $\lambda$-tracial subspace representation and the corresponding von Neumann algebra is finite-dimensional, then $G$ has a $(d,r)$-subspace representation with $d/r\leq\lambda$.
\end{proposition}
\begin{proof}
Let $(\cH,\cM, \{E_g\},\tau)$ be a $\lambda$-tracial subspace representation of $G$, where $\cM$ is finite-dimensional. Thus $\cM$ is ($*$-isomorphic to) a direct-sum of finite-dimensional matrix algebras: $\cA=\oplus_{i=1}^\ell M(d_i,\C)$. For each projection $E_g\in \cA$, it can be uniquely represented as $E_g=\oplus_{i=1}^\ell E_g^i$, where $E_g^i\in M(d_i,\C)$ is a projection for each $i\in[\ell]$. For the tracial state $\tau:\cA\to\C$, there exists a probability distribution $p=(p_1,\dots,p_\ell)\in [0,1]^\ell$ such that $\tau(\oplus_{i=1}^\ell X_i)=\sum_{i=1}^\ell \frac{p_i}{d_i}\Tr(X_i)$.
% Let $q=(q_1,\dots,q_\ell)$ such that $q_i=\frac{p_i}{d_i}$ for each $i\in[\ell]$. Then $\sum_{i=1}^\ell q_id_i=1$.
Note that since $\tau(E_g)=1/\lambda$ for every $g\in V(G)$, we have that
\[
\frac{1}{\lambda}=\sum_{i=1}^\ell \frac{p_i}{d_i}\Tr(E_g^i)=\sum_{i=1}^\ell q_i \, \rk(E_g^i) \text{ for all } g \in V(G),
\]
where we set $q_i = p_i/d_i$ for each $i \in [\ell]$.
Let
\begin{equation}\label{eq: lp in finite-dim C*-algebra}
% \begin{split}
% t^*=\max\{t~\text{ s.t. }~&t \in [0,1],\ p\in[0,1]^\ell,\ \sum_{i=1}^\ell p_i = 1, \\
% &\sum_{i=1}^\ell \frac{p_i}{d_i}\rk(E_g^i)=t~\forall\ g\in V(G)\}.
% \end{split}
\begin{split}
t^*=\max\ t~\text{ s.t. }~&t \in [0,1],\ q \in \R^\ell_{\geq 0},\ \sum_{i=1}^\ell d_i q_i = 1, \\
&\sum_{i=1}^\ell q_i \, \rk(E_g^i)=t \quad \forall\ g\in V(G).
\end{split} 
% \begin{split}
% t^*=\max\{ t :\exists\  q_1,\cdots, q_l \ge 0\  s.t.\  \text{ for all }g\in V(G),\
% &\sum_{i=1}^\ell q_i \, \rk(E_g^i)=t, \text{ and } \ \sum_{i=1}^\ell d_i q_i = 1\}
% \end{split}
\end{equation}
It is clear that $t^*\geq 1/\lambda$. Since the coefficients of the constraints in the linear program described in \cref{eq: lp in finite-dim C*-algebra} are all integers, the optimal value $t^*$ is attained by a vector $(t^*,q_1,\dots,q_\ell)$ of rational numbers. Write $q_i=m_i/d$ for some integers $m_i$ ($i \in [\ell]$) and $d$. Let $E_g'=\oplus_{i=1}^\ell E_g^i\otimes I_{m_i}$ for each $g\in V(G)$ and observe that $E_g' \in M_d$ since $\sum_i d_i m_i = \sum_i d_i q_i d = d$. Note that
\[
\rk(E_g')=\sum_{i=1}^\ell m_i\rk(E_g^i)=dt^*.
\]
Thus if we can show $\ran(E_g')\cap(\sum_{g'\in N_{\overline{G}}(g)}\ran(E'_{g'}))=\{0\}$ for every $g\in V(G)$, then we have obtained a $(d,dt^*)$-subspace representation of $G$, where $d/(dt^*)=1/t^*\leq\lambda$.

We are only left to prove that $\ran(E_g')\cap(\sum_{g'\in N_{\overline{G}}(g)}\ran(E'_{g'}))=\{0\}$ (the closure can be omitted since these projections are finite dimensional). Note that for each $g\in V(G)$,
\[\ran(E_g')=\oplus_{i=1}^\ell \ran(E_g^i)\otimes \C^{m_i}\subseteq (\oplus_{i=1}^\ell \ran(E_g^i))\otimes \C^{m}=\ran(E_g)\otimes \C^m,\]
where $m\geq m_i$ for all $i\in[\ell]$
Since $\ran(E_g)\cap(\sum_{g'\in N_{\overline{G}}(g)}\ran(E_{g'}))=\{0\}$, we have
\[\ran(E_g)\otimes\C^m\cap(\sum_{g'\in N_{\overline{G}}(g)}\ran(E'_{g'})\otimes\C^m)=\{0\}.\]
This concludes the proof.
\end{proof}

As a corollary, we show that if we further restrict the von Neumann algebra in a tracial subspace representation to be commutative, we can get a fractional clique cover (or equivalently, a fractional coloring of the complement).
\begin{corollary}\label{cor: abelian C*-algebra)}
If a graph $G$ has a $\lambda$-tracial subspace representation and the corresponding von Neumann algebra is commutative, then $\overline{G}$ has a $d:r$-coloring with $d/r\leq\lambda$.
\end{corollary}
\begin{proof}
Assume that $(\cH,\cM, \{E_g\},\tau)$ is a $\lambda$-tracial subspace representation and $\cM$ is commutative. Note that the $*$-algebra $\cM_0$ generated by the finite family of mutually commuting projection $\{E_g\}$ is always finite dimensional, since it is spanned by the monomials of the form $\prod_{g\in V(G)} E_g^{\epsilon_g}$, where  $\epsilon_g\in\{0,1\}$ for each $g\in V(G)$. Thus we may assume $\cM$ is finite-dimensional otherwise replace $\cM$ by the subalgebra $\cM_0$.

By the proof of~\cref{prop: finite dim C^* algebra}, we can assume that $\{\ran(E_g)\}$ forms a $d/r\leq\lambda$ subspace representation. For commuting operators $E,F \in \cM$ we have $E \wedge F = EF$. Hence
\[E_g E_{N_{\overline{G}}(g)}=E_g\wedge E_{N_{\overline{G}}(g)}=0.\]
Using that $E_{g'}\le E_{N_{\overline{G}}(g)}$ for every $g'\in N_{\overline{G}}(g)$, this implies $E_gE_{g'}=0$ for any $\{g,g'\}\in E(\overline{G})$. This gives a commutative feasible solution of the projective rank, which can be transformed into a $d:r$-coloring of $\overline{G}$ using the proof of~\cite[Theorem 6.8]{PAULSEN2016}.
\end{proof}

We show that the tracial Haemers bound is upper bounded by the tracial rank:
\begin{proposition}\label{prop: upper bound tracial rank}
Let $G$ be a graph. Then $\cH_{tr}(G)\leq\xitr(G)$.
\end{proposition}
\begin{proof}
Let $(\cA,\{E_g\},\tau)$ be a feasible solution of $\xitr(G)$ with value $\lambda$, as given in~\cref{def: tracial rank}, where $\tau$ is a faithful tracial vector state. Taking the GNS construction $\pi: \cA\to B(\cH_\pi)$ of $\tau$, $\tau$ induces a tracial vector state $\phi_\tau$ on the $C^*$-algebra $\pi(\cA)$, which natural extends to a normal tracial state on the von Neumann algebra $\cM=\overline{\pi(\cA)^{\text{sot}}}$. Moreover, since $E_gE_{g'}=0$ for all $\{g,g'\}\in E(\overline{G})$,
 we have $\pi(E_g)\pi(E_{g'})=0,$ which implies $\pi(E_g)\pi(E_{N_{\overline{G}}(g)})=0$ (and thus in particular $\pi(E_g) \wedge \pi(E_{N_{\overline G}(g)}) = 0$). This shows that  $(\cH_\pi,\cM, \pi(E_g),\phi_\tau)$ is a tracial subspace representation of $G$ with value $\lambda$.
\end{proof}

As a first example, we compute the tracial Haemers bound of odd cycles. To do so, we follow the same strategy that has been used to show the analogous result for the fractional Haemers bound in~\cite{bukh2018fractional}.
\begin{example}\label{ex: tracial Haemers odd cycle}
For all $d\in\N$, $\cH_{tr}(C_{2d+1})=d+\frac{1}{2}$.
\end{example}
\begin{proof}
Note that $\cH_{tr}(C_{2d+1})\leq \cH_f(C_{2d+1})=d+\frac{1}{2}$. Hence, we only need to show the lower bound.
Let $(\cH,\cM,\{E_i\}_{i\in\Z_{2d+1}},\tau)$ be a $\lambda$-tracial subspace representation of $C_{2d+1}$. First observe that we have $(E_0\wedge E_1) \vee (E_0\wedge E_{2d})\leq E_0$ and using the graph structure of $C_{2d+1}$ we have $E_1\wedge E_{2d}=0$ since $\{1,2d\}\in E(\overline{C_{2d+1}})$. \cref{thm:rank} (i) then implies that
\begin{equation} \label{eq: first decomp}
\frac{1}{\lambda} = \tau(E_0) \geq \tau(E_0\wedge E_1)+\tau(E_0\wedge E_{2d}).
\end{equation}
We proceed to lower bound $\tau(E_0\wedge E_1)$ and $\tau(E_0 \wedge E_{2d})$. Note that both are of the form $\tau(E_i \wedge E_{i+1 \bmod 2d+1})$. Using the symmetry of the odd cycle, it thus suffices to lower bound $\tau(E_0 \wedge E_1)$.

For any $k \in \N$ for which  $2 \leq  k \leq d -1$, let $I_k$ be the independent set $I_k=\{3,5,7,\dots,2k-1\}$. Let $E_{\{0,1\}}=E_0\vee E_1$ be the projection onto $\cl(\ran(E_0)+\ran(E_1))$ and $E_{I_k}=\bigvee_{i\in I_k}E_i$ be the projection onto $\cl(\sum_{i\in I_k}\ran(E_i))$. Since the edge $\{0,1\}$ is not adjacent to any vertex in $I_{d-1}$, we have  $\tau(E_{\{0,1\}}\vee E_{I_{d-1}})\leq 1$. We can expand the left hand side using \cref{thm:rank}~(i) (iteratively) to show that
\begin{equation} \label{eq: chain indep set}
1 \geq \tau(E_{\{0,1\}}\vee E_{I_{d-1}})=\tau(E_{\{0,1\}})+\sum_{i\in I_{d-1}} \tau(E_i)=\tau(E_0)+\tau(E_1)-\tau(E_0\wedge E_1)+\sum_{i\in I_{d-1}}\tau(E_i).
\end{equation}
Here the first equality uses the fact that for every $2\leq k\leq d-1$ we have $\tau(E_{\{0,1\}}\vee E_{I_k})=\tau(E_{\{0,1\}}\vee E_{I_{k-1}})+\tau(E_{2k-1})$, since vertex $2k-1$ is not adjacent with any vertices in $I_{k-1}\cup\{0,1\}$. Combining \cref{eq: chain indep set} with the identity $\tau(E_i)=1/\lambda$ for all $i\in\Z_{2d+1}$, shows
\[\tau(E_0\wedge E_1)\geq \frac{d+1}{\lambda}-1.
\]
We can use this lower bound in \cref{eq: first decomp} to obtain
\[
\frac{1}{\lambda}=\tau(E_0)\geq \tau(E_0\wedge E_1+E_0\wedge E_{2d})=\tau(E_0\wedge E_1)+\tau(E_0\wedge E_{2d})\geq \frac{2d+2}{\lambda}-2.
\]
Multiplying both sides by $\lambda$ and rearranging gives $\lambda\geq d+\frac{1}{2}$.
\end{proof}
\begin{corollary}
The tracial Haemers bound and the Lov\'asz theta function are incomparable.
\end{corollary}
\begin{proof}
The previous example shows that $\cH_{tr}(C_5)=2.5>\sqrt{5}=\vartheta(C_5)$. On the other hand, let $G$ be the complement of the Schl\"afli graph, $\cH_{tr}(G)\leq\cH_f(G;\C)\leq \cH(G;\C)\leq 7<9=\vartheta(G)$, where the second to last inequality was proved in~\cite{haemers1979some}.
\end{proof}

In the following subsections, we proceed to establish the monotonicity, additivity, normalization and multiplicativity of the tracial Haemers bound, along with some equivalent definitions. These enable us to conclude that $\cH_{tr}\in \bX(\cG,\leq_{qc})$ and for any graph $G$,
\[
\alpha_{qc}(G)\leq\Theta_{qc}(G)\leq\cH_{tr}(G).
\]

\subsection{Monotonicity of the tracial Haemers bound}
We first introduce an equivalent formulation of a tracial subspace representation of a graph using two subspaces per vertex of $G$.
\begin{proposition}\label{prop: first relaxation of tracial Haemers bound}
Let $G$ be a graph. Then $G$ has a $\lambda$-tracial subspace representation if and only if there exist a Hilbert space $\cH$, a von Neumann algebra $\cM\subseteq B(\cH)$ with a normal tracial state $\tau:\cM \to \C$, containing projections $E_g$ and $F_g$ for all $g\in V(G)$, such that
\begin{enumerate}[label=\upshape(\roman*)]
\item $\ran(E_g)\cap \ran(F_g)^\perp=\{0\}$ for any $g\in V(G)$,
\item $F_gE_{g'}=0$ for any $\{g,g'\}\in E(\overline{G})$,
\item $\tau(F_g)\geq \tau(E_g) \geq \frac{1}{\lambda}$ for all $g\in V(G)$.
\end{enumerate}
\end{proposition}
\begin{proof}
Let $(\cH,\cM,\{E_g\},\{F_g\},\tau)$ be given as in \cref{prop: first relaxation of tracial Haemers bound}. Let $S_g=\ran(E_g)$ and $T_g=\ran(F_g)$ for any $g\in V(G)$. We prove that $S_g\cap\cl(\sum_{g'\in N_{\overline{G}}(g)} S_{g'})=\{0\}$, which implies that $(\cH, \hat{\cM},\{E_g\},\tau)$ is a $\lambda$-tracial subspace representation, where $\hat{\cM}$ is the von Neumann algebra generated by the set of projections $\{E_g\}$. Assume $x\in S_g\cap\cl(\sum_{g'\in N_{\overline{G}}(g)} S_{g'})$, we show $x=0$. To see this, note that $E_{g'} F_{g}=0$ for any $\{g,g'\}\in E(\overline{G})$. Thus $T_g \perp S_{g'}$ for any $\{g,g'\}\in E(\overline{G})$ and therefore $T_g\perp\cl(\sum_{g'\in N_{\overline{G}}(g)} S_{g'})$. This shows that $\cl(\sum_{g'\in N_{\overline{G}}(g)} S_{g'}) \subseteq T_g^\perp$ and hence $x \in T_g^\perp$. Also, by assumption $x \in S_g$. However $S_g \cap T_g^\perp = \{0\}$, so $x = 0$.

Let $(\cH,\cM,\{E_g\},\tau)$ be a $\lambda$-tracial subspace representation as given in \cref{def: tracial subspace rep}. For each $g\in V(G)$, let $F_g=I_\cM-(\bigvee_{g\in V(G)}E_g)$.  Then $F_g$ belongs to $\cM$ and we have $\ran(E_g)\cap \ran(F_g)^\perp=\{0\}$ by construction. Moreover we have $F_g E_{g'}=0$ for any $\{g,g'\}\in E(\overline{G})$ since the range of $E_{g'}$ is orthogonal to the range of $F_g$. It remains to show that $\tau(F_g)\geq \tau(E_g)$ for every $g\in V(G)$.
Since $E_g \wedge (I-F_g) = 0$, \cref{thm:rank} (ii) then shows that $\tau(F_g) \geq \tau(E_g)$ for any $g\in V(G)$.
\end{proof}

We are ready to prove that $\cH_{tr}$ is $\leq_{qc}$-monotone.
\begin{proposition}\label{prop: monotone tracial Haemers}
For graphs $G,H\in\cG$, we have that $G\leq_{qc} H$ implies $\cH_{tr}(G)\leq\cH_{tr}(H)$.
\end{proposition}
\begin{proof}
Let $(\cH_H,\cM_H,\{E_h\},\{F_h\},\tau_H)$ be a $\lambda$-tracial subspace representation of $H$ as given in~\cref{prop: first relaxation of tracial Haemers bound}. Let $\cH_{G,H}$ be a Hilbert space,
and let $\cA(G,H)\subseteq B(\cH_{G,H})$ be a $C^*$-algebra containing projections $\{P_{g,h}:~g\in V(G),h\in V(H)\}$, with tracial state $\tau_{G,H}:\cA(G,H) \to \C$, which is a feasible solution of $G \leq_{qc} H$. We first show that we can assume that $\cA(G,H)$ is also a von Neumann algebra. Indeed, we can take the GNS representation $\pi$ of $\cA(G,H)$ with respect to $\tau_{G,H}$, and the von Neumann algebra $\cM(G,H)$ generated by $\pi(\cA(G,H))$, the image of the GNS construction. Then the tracial state $\tau'$ naturally extends to $\cM(G,H)$. For ease of notation, we also use $\tau_{G,H}$ to denote the tracial state on $\cM(G,H)$ and $\{P_{g,h}:~g\in V(G),h\in V(H)\}$ to denote their images in $\cM(G,H)$.

We construct a $\lambda$-tracial subspace representation of $G$. Set $A_g=\sum_{h\in V(H)}P_{g,h}\otimes E_h$ and $B_g=\sum_{h\in V(H)}P_{g,h}\otimes F_h$ for all $g\in V(G)$ and let $\cM_G$ be the von Neumann algebra generated by $\{A_g\}$ and $\{B_g\}$, which is a subalgebra of $\cM(G,H)\overline{\otimes}\cM_H$. Let $\tau_G=\tau_{G,H}\otimes\tau_H$ be the product tracial state on $\cM(G,H)\overline{\otimes}\cM_H$, thus also a tracial state on $\cM_G$.
We first verify that $A_g$ and $B_g$ are projections. Indeed,
\[
A_g^2=\sum_{h,h'\in V(H)} P_{g,h}P_{g,h'}\otimes E_hE_{h'}=\sum_{h\in V(H)}P_{g,h}\otimes E_h=A_g,
\]
where the second equality uses the fact that $P_{g,h}P_{g,h'}=0$ for all $h\neq h'$. The same argument works for $B_g$ as well. We then lower bound the trace of $A_g$ and $B_g$. Note that
\[
\tau_G(A_g)=\sum_{h\in V(H)}\tau_{G,H}(P_{g,h})\otimes \tau_H(E_h)\geq \frac{1}{\lambda}\sum_{h\in V(H)}\tau_{G,H}(P_{g,h})=\frac{1}{\lambda},
\]
where the last equality holds since $\sum_{h\in V(H)}P_{g,h}=I_{\cM(G,H)}$. On the other hand,
\[
\tau_G(B_g)=\sum_{h\in V(H)}\tau_{G,H}(P_{g,h})\otimes \tau_H(F_h)\geq\sum_{h\in V(H)}\tau_{G,H}(P_{g,h})\otimes \tau_H(E_h)=\tau_G(A_g).
\]

We now verify that $B_gA_{g'}=0$ if $\{g,g'\}\in E(\overline{G})$. Note that
\[
B_gA_{g'}=\sum_{h,h'\in V(H)} P_{g,h}P_{g',h'}\otimes F_hE_{h'}=\sum_{\{h,h'\}\in E(\overline{H})}P_{g,h}P_{g',h'}\otimes F_hE_{h'}=0,
\]
where the second equality holds since $P_{g,h}P_{g',h'}=0$ if $\{g,g'\}\in E(\overline{G})$ and $\{h,h'\}\in E(H)$ or $h=h'$ and the third equality holds since $F_hE_{h'}=0$ if $\{h,h'\}\in E(\overline{H})$.

We are left to show that $\ran(A_g)\cap\ran(B_g)^\perp=\{0\}$ for every $g\in V(G)$. We first observe that
\[
\ran(A_g)=\oplus_{h\in V(H)}\ran(P_{g,h}\otimes E_h)~\text{and}~\ran(B_g)=\oplus_{h\in V(H)}\ran(P_{g,h}\otimes F_h),
\]
since $(P_{g,h}\otimes E_h)(P_{g,h'}\otimes E_{h'})=P_{g,h}P_{g,h'}\otimes E_hE_{h'}=0$ and $(P_{g,h}\otimes F_h)(P_{g,h'}\otimes F_{h'})=P_{g,h}P_{g,h'}\otimes F_hF_{h'}=0$. Assume $x\in\ran(A_g)\cap\ran(B_g)^\perp$, that is, assume $A_gx=x$ and $B_gx=0$. Note that
\[
A_gB_gA_g=\sum_{h_1,h_2,h_3\in V(H)}P_{g,h_1}P_{g,h_2}P_{g,h_3}\otimes E_{h_1}F_{h_2}E_{h_3}=\sum_{h\in V(H)}P_{g,h}\otimes E_{h}F_{h}E_{h},
\]
where the second equality use the fact that $P_{g,h}$ is a projection and $P_{g,h}P_{g,h'}=0$ for any $h\neq h'$.
 Thus,
\[
0=\langle x, B_gx\rangle=\langle x,A_gB_gA_gx\rangle=\sum_{h\in V(H)} \langle x,(P_{g,h}\otimes E_{h}F_{h}E_{h})x\rangle.
\]
This implies that $\langle x,(P_{g,h}\otimes E_{h}F_{h}E_{h})x\rangle=0$ for any $h\in V(H)$ as $P_{g,h}\otimes E_{h}F_{h}E_{h}$ is a positive operator. Let $x_h=(P_{g,h}\otimes E_{h})x\in\ran(P_{g,h})\otimes \ran(E_h)$. We shall prove that $x_h=0$ for all $h\in V(H)$, then $x=0$ and we are done. For every fixed $h$, write the Schmidt decomposition of $x_h$ as $\sum_{i}\mu_iy_i\otimes z_i$, where $\mu_i>0$, $y_i\in\ran(P_{g,h})$ and $z_i\in \ran(E_h)$ for every $i$. Then the identity $\langle x,(P_{g,h}\otimes E_{h}F_{h}E_{h})x\rangle=0$ further implies
\[
0=\langle x_h, (P_{g,h}\otimes F_{h})x_h\rangle=\sum_{i,j}\mu_i\langle y_i, P_{g,h}y_j\rangle\langle z_i, F_hz_j\rangle=\sum_i\mu_i\langle z_i,F_hz_i\rangle,
\]
where the last equality holds since $P_{g,h}y_i=y_i$ for every $i$ and $\langle y_i,y_j\rangle=0$ for every $i\neq j$. Thus, $F_hz_i=0$ for every $i$. Note that $\ran(E_h)\cap\ran(F_h)^\perp=\{0\}$, which implies $z_i=0$ for all $i$. This concludes our proof.
\end{proof}

\subsection{Additivity of the tracial Haemers bound}
In this subsection, we show additivity of $\cH_{tr}(G)$. The subadditivity follows naturally from the original definition of tracial Haemers bound.

\begin{proposition} \label{prop: subadditive}
For graphs $G$ and $H$ we have $\cH_{tr}(G \sqcup H) \leq \cH_{tr}(G) +\cH_{tr}(H)$.
\end{proposition}
\begin{proof}
Let $(\cH_G,\cM_G,\{E_g\},\{F_g\},\tau_G)$ and $(\cH_H,\cM_H,\{\tilde E_h\},\{\tilde F_h\},\tau_H)$ be $\lambda_G$-tracial subspace representation of $G$ and $\lambda_H$-tracial subspace representation of $H$, respectively, given in terms of~\cref{prop: first relaxation of tracial Haemers bound}. We construct a $(\lambda_G+\lambda_H)$-tracial subspace representation of $G \sqcup H$. For $g \in V(G)$ define $A_g = (E_g,0)$ and $B_g=(F_g,0)$ and for $h \in V(H)$ let $A_h =(0,\tilde E_h)$ and $B_h=(0,\tilde F_h)$. Let $\cH=\cH_G\oplus\cH_H$ and $\cM$ be the von Neumann algebra generated by $\{A_v:v\in V(G)\cup V(H)\}$ and $\{B_v:v\in V(G)\cup V(H)\}$. Note that $\cM\subseteq\cM_G\oplus\cM_H$. Finally, for an element $(X,Y) \in \cM$, define $\tau((X,Y)) = \frac{\lambda_G}{\lambda_G+\lambda_H}\tau_G(X) + \frac{\lambda_H}{\lambda_G + \lambda_H} \tau_H(Y)$. We verify that $\tau$ is a tracial state on $\cM$. Note that $\tau((X^*X,Y^*Y)=\frac{\lambda_G}{\lambda_G+\lambda_H}\tau_G(X^*X) + \frac{\lambda_H}{\lambda_G + \lambda_H} \tau_H(Y^*Y)\geq 0$, $\tau((I_{\cM_G},I_{\cM_H}))=1$ and
\[
\begin{split}
\tau(X_1X_2,Y_1Y_2)&=\frac{\lambda_G}{\lambda_G+\lambda_H}\tau_G(X_1X_2) + \frac{\lambda_H}{\lambda_G + \lambda_H} \tau_H(Y_1Y_2)\\
&=\frac{\lambda_G}{\lambda_G+\lambda_H}\tau_G(X_2X_1) + \frac{\lambda_H}{\lambda_G + \lambda_H} \tau_H(Y_2Y_1)=\tau(X_2X_1,Y_2Y_1).
\end{split}
\]
We show that $(\cH,\cM,\{A_v\}_{v\in V(G)\cup V(H)},\{B_u\}_{u\in V(G)\cup V(H)},\tau)$ is a $(\lambda_G+\lambda_H)$-tracial subspace representation of $G\sqcup H$. First note that for $u \in V(G)$ we have $\tau(A_u) = \frac{\lambda_G}{\lambda_G+\lambda_H} \tau_G(E_u) + 0 \geq \frac{1}{\lambda_G+\lambda_H}$ and similarly for $v \in V(H)$, we have $\tau(A_v) \geq 0 + \frac{\lambda_H}{\lambda_G+\lambda_H} \tau_H(\tilde E_v) \geq \frac{1}{\lambda_G+\lambda_H}$. It is also straightforward that $\tau(B_v)\geq \tau(A_v)$ for every $v\in V(G)\cup V(H)$. We then verify that $B_vA_u=0$ for $\{v,u\}\in E(\overline{G\sqcup H})$. If $v\in V(G)$ and $u\in V(H)$, $B_vA_u=(F_v,0)(0,\tilde E_u)=0$. A similar argument holds for $v\in V(H)$ and $u\in V(G)$. If $v,u\in V(G)$ and $\{v,u\}\in E(\overline{G})$, $B_vA_u=(F_v,0)(E_u,0)=(F_vE_u,0)=0$. A similar argument holds for $v,u\in V(H)$ and $\{v,u\}\in E(\overline{H})$. Finally we verify that $\ran(A_u)\cap\ran(B_u)^\perp=\{0\}$. We do so for $u\in V(G)$ and the argument for $u\in V(H)$ is similar. Note that $\ran(A_u)=(\ran(E_u), 0)$ and $\ran(B_u)=(\ran(F_u), 0)$. Then $\ran(B_u)^\perp=(\ran(F_u)^\perp,B(\cH_H))$. Then $\ran(A_u)\cap\ran(B_u)^\perp=\{0\}$ and $\ran(E_u)\cap\ran(F_u)^\perp=\{0\}$. This concludes the proof.
\end{proof}

To show superadditivity, we first strengthen the notion of a subspace representation to that of a projection representation by replacing the subspace intersection condition by a norm-inequality condition.
\begin{definition}[Tracial projection representation] \label{def: tracial projection rep}
We say a graph $G$ has a \emph{$\lambda$-tracial projection representation} if a von Neumann algebra $\cM$ containing projections $E_g$ and $F_g$ for all $g\in V(G)$, with a normal tracial state $\tau:\cA \to \C$, such that
\begin{enumerate}[label=\upshape(\roman*)]
\item $\norm{E_g(I_\cA-F_{g})}<1$ for any $g\in V(G)$,
\item $F_gE_{g'}=0$ for any $\{g,g'\}\in E(\overline{G})$,
\item $\tau(F_g)\geq \tau(E_g) \geq \frac{1}{\lambda}$ for all $g\in V(G)$.
\end{enumerate}
\end{definition}
As for \cref{def: tracial subspace rep}, the Hilbert $\cH$ in the above definition is unnecessary but we impose it for the convenience of discussion.
Note that condition (i) above is a strengthening of the projection-intersection condition. Indeed, for two projections $P,Q\in B(\cH)$, if $\norm{PQ}<1$, then $P\wedge Q=0$. If the underlying Hilbert space is finite dimensional, then the reverse direction also holds. In the infinite-dimensional setting, the reverse direction does not necessarily hold.  Nevertheless, we shall prove that, for any graph $G$, the infimum over tracial subspace representations and the infimum over tracial projection representation are the same.
\begin{proposition}\label{prop: first relaxation of tracial Haemers bound 2}
Let $G$ be a graph. Then
\[
\cH_{tr}(G)=\inf\{\lambda:~G~\text{has a}~\lambda\text{-tracial projection representation}\}.
\]
\end{proposition}
\begin{proof}
We show equality between $\cH_{tr}(G)$ and the parameter defined in this proposition. Note that the condition $\|E_g(I_\cM - F_g)\|<1$ implies $E_g\wedge (I_\cM - F_g)= 0$. This shows that $\cH_{tr}(G)$ is at most the parameter defined in this proposition. To prove the converse inequality, consider a tuple $(\cH,\cM,\{E_g\},\{F_g\},\tau)$ of $G$ satisfying \cref{prop: first relaxation of tracial Haemers bound}. Let $0< \eps \leq 1$. \cref{thm:rank} (iii) implies that for every $g\in V(G)$, there is a projection $E_g^\eps\in\cM$ such that $E_g^\eps \leq E_g$, $\norm{E_g^\eps(I_\cM-F_g)}<1$ and $\tau(E_g^\eps)\ge \tau(E_g)-\epsilon\ge \frac{1}{\lambda}- \epsilon$. Let $\hat{\cM}$ be the von Neumann algebra generated by $\{E_g^\eps\}$ and $\{F_g\}$. It follows that $(\cH,\hat{\cM},\{E^\eps_g\},\{F_g\},\tau)$ is a $\frac{\la}{1-\eps\la}$-tracial projection representation of $G$. Letting $\eps$ tend to $0$ and taking the infimum over $\lambda$ concludes the proof.
\end{proof}

With \cref{prop: first relaxation of tracial Haemers bound 2}, we can further restrict the Hilbert space $\cH$ in \cref{prop: first relaxation of tracial Haemers bound}, which will play a crucial rule in the proof of superadditivity (and the proof of supermultiplicativity).
\begin{proposition}\label{prop: hilbert space restriction}
Let $G$ be a graph. If there exists a $\lambda$-tracial projection representation of $G$, then there exists a $\lambda$-tracial subspace representation $(\cH,\cM,\{E_g\},\{F_g\},\tau)$ in the sense of \cref{prop: first relaxation of tracial Haemers bound}, where
\[
\cH =\ran(\bigvee_{g \in V(G)} F_g).
% = \cl(\sum_{g \in V(G)} \ran(F_g)).
\]

\end{proposition}
\begin{proof}
Let $(\cH,\cM,\{E_g\},\{F_g\},\tau)$ be a $\lambda$-tracial projection representation of $G$ and
let $F$ be the projection onto $\cl(\sum_{g\in V(G)}\ran(F_g))$. Let $E'_g$ be the left support projection of $F E_g$ (i.e.~the projection onto the closure of $\ran(F E_g)$) and set $F'_g=F_g$.
Using the inclusions $\ran(E'_g),\ran(F'_g)\subseteq \ran(F)$, we observe that $F$ commutes with both $F_g'$ and $E_g'$ and that it acts as the identity on the algebra generated by the $E_g'$'s and $F_g'$'s.
Let $\cH'=\ran(F)$ and let $\hat{\cM}$ be the von Neumann algebra generated by $\{E'_g\}$ and $\{F'_g\}$, where $F$ is the identity element. Define $\tau':\hat{\cM}\to\C$ as $\tau'(X)=\tau(X)/\tau(F)$ for $X\in\hat{\cM}$. This is a tracial state on $\hat{\cM}$. We show that $(\cH',\hat{\cM},\{E'_g\},\{F'_g\},\tau')$ satisfies \cref{prop: first relaxation of tracial Haemers bound} with the same value.

We first verify that $F'_g E'_{g'}
=0$ for every $\{g,g'\}\in E(\overline{G})$. We have
\[
F_g F E_{g'}=F_gE_{g'}=0 \quad \text{for any $\{g,g'\}\in E(\overline{G})$}.
\]
Using that $E_g'$ is the left support projection on $FE_g$ (in particular the property in \cref{eq: left and right support projection}), this indeed implies that $F_g' E_{g'}' = F_g E_{g'}' = 0$.
Next we show $\ran(E'_g)\cap\ran(F_g)^\perp=\{0\}$. Let $v\in \ran(E'_g)\cap\ran(F_g)^\perp$, we show $v=0$. By the definition of $E'_g$, there exists a sequence $\{u_i\in\cH\}_{i\in\N}$ such that $FE_g u_i\to v$ as $i\to \infty$. On the other hand, $v\in\ran(F_g)^\perp$ implies that $F_gv=0$. Thus,
\[
0=F_gv=F_g\lim_{i\to \infty }FE_gu_i=\lim_{i\to \infty }F_gFE_gu_i=\lim_{i\to \infty }F_gE_gu_i.
\]
In other words, $\|F_g E_g u_i\| \to 0$ as $i \to \infty$. We show this implies $\|E_g u_i\| \to 0$.
Let $\norm{E_g(I_\cA-F_g)}=\delta_g<1$ for any $g\in V(G)$. Then by the triangle inequality, for any $i\in\N$,
\[
\norm{F_g E_g u_i}\ge \norm{E_g u_i }-\norm{(I_\cA-F_g)E_g u_i }\ge \norm{E_g u_i }-\norm{E_g(I_\cA-F_g)} \norm{E_gu_i}=(1-\delta_g)\norm{E_gu_i}.
\]
Thus,
\begin{align*}
\lim_{i\to \infty}\norm{E_gu_i}\le (1-\delta_g)^{-1}\lim_{i\to \infty}\norm{F_gE_gu_i} =0,
\end{align*}
which implies $v=\lim_{i\to\infty}FE_gu_i=0$. (The second inequality uses the submultiplicativity of the operator norm.)

Finally we show $\tau'(E'_g)=\tau(E'_g)/\tau(F)\geq\frac{1}{\lambda}$. Since $\tau(F) \leq 1$, it suffices to show $\tau(E'_g)=\tau(E_g)$.
Since $\ran(F)^\perp\subseteq\ran(F_g)^\perp$, we have $\ran(E_g)\cap\ran(F)^\perp\subseteq\ran(E_g)\cap\ran(F_g)^\perp=\{0\}$.  Applying \cref{thm:rank} (ii) with $Q=F$ and $P=E_g$ therefore shows that for any tracial state $\tau:\cM\to\C$ we have $\tau(E_g')=\tau(E_g)$.
\end{proof}

We are now ready to prove superadditivity of tracial Haemers bound.
\begin{proposition} \label{prop: superadditive}
For graphs $G$ and $H$ we have $\cH_{tr}(G \sqcup H) \geq \cH_{tr}(G) +\cH_{tr}(H)$.
\end{proposition}
\begin{proof}
Let $(\hat{\cH},\hat{\cM},\{\hat{E}_v\},\{\hat{F}_v\},\hat{\tau})$ be a $\lambda$-tracial projection representation of $G \sqcup H$ as formulated in \cref{def: tracial projection rep}. \cref{prop: hilbert space restriction} shows that we can obtain a $\lambda$-tracial subspace representation $(\cH,\cM,\{E_v\},\{F_v\},\tau)$ with $\cH = \ran(F)$, where $F$ is the projection onto the closure of $\sum_{v\in V(G)\cup V(H)}\ran(F_v)$. Using the graph structure of $G\sqcup H$ we see that the projections $\{E_g\}_{g\in V(G)}$ and $\{F_g\}_{g\in V(G)}$ satisfy
\begin{itemize}
\item[(1)] $\ran(E_g)\cap \ran(F_g)^\perp=\{0\}$ for any $g\in V(G)$,
\item[(2)] $F_gE_{g'}=0$ for any $\{g,g'\}\in E(\overline{G})$.
\end{itemize}
Similar conditions hold for $\{E_h\}_{h\in V(H)}$ and $\{F_h\}_{h\in V(H)}$. Let $E_G$ be the projection onto the closure of $\sum_{g\in V(G)}\ran(E_g)$ and similarly let $E_H$ be the projection onto the closure of $\sum_{h\in V(H)}\ran(E_h)$. Let $\cH_G=\ran(E_G)$ and $\cH_H=\ran(E_H)$. Let $\cM_G$ (resp.~$\cM_H$) be the von Neumann algebra generated by the projections $\{E_g\}_{g\in V(G)}$ with identity element $E_G$ (resp.~$\{E_h\}_{h\in V(G)}$ with identity $E_H$). Thus $\cM_G\subseteq B(\cH_G)$ and $\cM_H\subseteq B(\cH_H)$. Let $\tau_G:\cM_G\to \C$ be defined as $\tau_G(X)=\tau(X)/\tau(E_G)$ for $X\in\cM_G$ and $\tau_H:\cM_H\to \C$ be defined as $\tau_H(Y)=\tau(Y)/\tau(E_H)$ for $Y\in\cM_H$. Since $\cM_G$ and $\cM_H$ are subsets of $\cM$, $\tau_G$ and $\tau_H$ are tracial states on $\cM_G$ and $\cM_H$ respectively. Analogously to the first part of the proof of \cref{prop: first relaxation of tracial Haemers bound} one can show that $(\cH_G,\cM_G,\{E_g\}_{g\in V(G)},\tau_G)$ and $(\cH_H,\cM_H,\{E_h\}_{h\in V(H)},\tau_H)$ are $(\tau(E_G)\lambda)$- and $(\tau(E_H)\lambda)$-tracial subspace representations of $G$ and $H$ respectively. It follows that $\cH_{tr}(G) + \cH_{tr}(H) \leq \lambda (\tau(E_G) + \tau(E_H))$. If we can show that $\tau(E_G)+\tau(E_H)\leq 1$, then we indeed obtain $\cH_{tr}(G \sqcup H) \geq \cH_{tr}(G) +\cH_{tr}(H)$.

To show that $\tau(E_G)+\tau(E_H)\leq 1$ we will apply \cref{thm:rank} to $E_G$ and $E_H$. To do so we are left to prove that $\ran(E_G)\cap \ran(E_H)=\{0\}$. Assume there is a nonzero vector $0\neq u\in\ran(E_G)\cap \ran(E_H)$. Note that for any $g\in V(G)$, $h\in V(H)$, we have $F_g E_h=0$ and $F_h E_g=0$, since there are no edges between vertices in $G$ and vertices in $H$. This shows that $F_h E_G=0$ for any $h\in V(H)$ and $F_g E_H=0$ for any $g\in V(G)$. Since $E_Gu=u$ and $E_Hu=u$, we have $F_vu=F_vE_Hu=0$ for every $v\in V(G)$ and $F_vu=F_vE_Gu=0$ for every $v\in V(H)$. Thus,
\[
u\in\bigcap_{g\in V(G)} \ran(F_g)^\perp\cap\bigcap_{h\in V(H)} \ran(F_h)^\perp=\bigcap_{v\in V(G)\cup V(H)} \ran(F_v)^\perp=\ran(F)^\perp=\{0\},
\]
where the last equality holds due to the assumption that $\cH=\ran(F)$.
\end{proof}

By the additivity, we can easily show that the tracial Haemers bound is normalized:
\begin{example}\label{ex: tracial Haemers empty graph}
Let $\overline{K_d}$ be the empty graph with $d$ vertices. Then $\cH_{tr}(\overline{K}_d)=d$.
\end{example}
\begin{proof}
It is easy to see that $\cH_{tr}(\overline{K}_1)=1$. Then $\cH_{tr}(\overline{K}_d)=\cH_{tr}(\sqcup_{i=1}^d\overline{K}_1)=\sum_{i=1}^d\cH_{tr}(\overline{K}_1)=d$.
\end{proof}
The normalization of the tracial Haemers bound, together with the monotonicity with respect to the commuting quantum homomorphism, implies that the tracial Haemers bound is an upper bound on the commuting quantum independence number.
\begin{corollary}
For every graph $G$, we have $\alpha_{qc}(G)\leq\cH_{tr}(G)$.
\end{corollary}
\begin{proof}
By definition, if $\alpha_{qc}(G)=d$, then $\overline{K_d}\leq_{qc} G$. The monotonicity and normalization of $\cH_{tr}$ together show that $d=\cH_{tr}(\overline{K_d})\leq\cH_{tr}(G)$.
\end{proof}

\subsection{Multiplicativity of the tracial Haemers bound}

We now discuss the multiplicativity of $\cH_{tr}(G)$. We first use \cref{prop: first relaxation of tracial Haemers bound} to prove that the tracial Haemers bound is submultiplicative.
\begin{proposition}\label{prop: tracial Haemers submultiplicative}
For any graphs $G$ and $H$, $\cH_{tr}(G\boxtimes H)\leq\cH_{tr}(G)\cH_{tr}(H)$.
\end{proposition}
\begin{proof}
Let $(\cH_G,\cM_G,\{E_g\},\{F_g\},\tau_G)$ and $(\cH_H,\cM_H,\{E'_h\},\{F'_h\},\tau_H)$ be tracial subspace representations of $G$ and $H$, respectively, in the formulation of \cref{prop: first relaxation of tracial Haemers bound}. Set $P_{g,h}=E_g\otimes E'_h$ and $Q_{g,h}=F_g\otimes F'_h$ for any $g\in V(G)$ and $h\in V(H)$. It is clear that $P_{g,h}$ and $Q_{g,h}$ are projections. Let $\cM_{G,H}$ be the von Neumann algebra generated by the projections $\{P_{g,h},Q_{g,h}:~g\in V(G),~h\in V(H)\}$. It follows that $\cM_{G,H}$ is a subalgebra of $\cM_G\otimes\cM_H$, the latter of which is equipped with the product tracial state $\tau_G\otimes\tau_H$. Thus $\tau=\tau_G\otimes\tau_H$ is a tracial state on $\cM_{G,H}$, and it is easy to see that $\tau(Q_{g,h})\geq \tau(P_{g,h})\geq \frac{1}{\lambda\mu}$ for any $g\in V(G)$ and $h\in V(H)$.

For $\{(g,h),(g',h')\}\in E(\overline{G\boxtimes H})$, $P_{g,h}Q_{g',h'}=E_g F_{g'}\otimes E'_hF'_{h'}=0$ because $\{g,g'\}\in E(\overline{G})$ or $\{h,h'\}\in E(\overline{H})$. This verifies the orthogonality conditions between $P_{g,h}$ and $Q_{g',h'}$.

We now show that $\ran(P_{g,h})\cap\ran(Q_{g,h})^\perp=\{0\}$ for every $g\in V(G)$ and $h\in V(H)$. Note that $\ran(P_{g,h})=\ran(E_g)\otimes\ran(E'_h)$ and
\[
\ran(Q_{g,h})^\perp=(\ran(F_g)\otimes\ran(F'_h))^\perp=\left(\cH_G\otimes \ran(F'_h)^\perp \right) \oplus \left(\ran(F_g)^\perp\otimes \ran(F'_h)\right).
\] Thus for $x\in\ran(P_{g,h})\cap \ran(Q_{g,h})^\perp$, we have $x=y+z$, where $x\in \ran(E_g)\otimes\ran(E'_h)$, $y\in \cH_G\otimes \ran(F'_h)^\perp$ and $z\in\ran(F_g)^\perp\otimes \ran(F'_h)$. For any linear functional $\cL:\ran(F'_h)\to\C$, we have $(I_G\otimes\cL)(x)\in\ran(E_g)$, $(I_G\otimes\cL)(y)=0$ and $(I_G\otimes\cL)(z)\in\ran(F_g)^\perp$. Thus we have $(I_G\otimes\cL)(x)=(I_G\otimes\cL)(z)\in\ran(E_g)\cap\ran(F_g)^\perp=\{0\}$. This implies $z=0$ because $\cL$ is arbitrary. Similarly, we have $y=0$ and hence $x=0$.
\end{proof}

To prove the supermultiplicativity of $\cH_{tr}(G)$, we need to relax the conditions in the $\lambda$-tracial subspace (resp.~projection) representation by modifying (i) and (iii) in \cref{prop: first relaxation of tracial Haemers bound} (resp.~\cref{prop: first relaxation of tracial Haemers bound 2}):
\begin{proposition}\label{prop: second relaxation of tracial Haemers bound}
Let $G$ be a graph. Then $G$ has a $\lambda$-tracial subspace representation if and only if there exist a Hilbert space $\cH$, a von Neumann algebra $\cM\subseteq B(\cH)$ containing projections $E_g$ and $F_g$ for all $g\in V(G)$, with a normal tracial state $\tau:\cM \to \C$, such that
\begin{enumerate}[label=\upshape(\roman*)]
\item $F_gE_{g'}=0$ for any $\{g,g'\}\in E(\overline{G})$,
\item $\tau(E_g)-\tau(X_g)\geq \frac{1}{\lambda}$ for all $g\in V(G)$, where $X_g$ is the projection onto the subspace $\ran(E_g)\cap \ran(F_g)^\perp$.
\end{enumerate}
\end{proposition}
\begin{proof}
We first prove that a tuple $(\cH,\cM, \{E_g\}, \{F_g\}, \tau)$ which satisfies \cref{prop: second relaxation of tracial Haemers bound} can be used to construct a $\lambda$-tracial subspace representation of $G$ in terms of~\cref{prop: first relaxation of tracial Haemers bound}. Let $\tilde{E}_g=E_g-X_g$, for every $g\in V(G)$, where $X_g$ is the projection onto the subspace $\ran(E_g)\cap \ran(F_g)^\perp$, and let $\tilde{F}_g=F_g$. Note that $\tilde{E}_g$ is a projection since $\ran(X_g) \subseteq \ran(E_g)$. Let $\tilde{\cM}$ be the von Neumann algebra generated by the projections $\tilde{E}_g$ and $\tilde{F}_g$ for all $g\in V(G)$. Then $\tilde{\tau}=\tau$ is a tracial state on $\tilde{\cM} \subseteq \cM$ and we observe that $\tilde{\tau}(\tilde{E}_g)=\tau(E_g)-\tau(X_g)\geq\frac{1}{\lambda}$.

It remains to show that (i) $\ran(\tilde{E}_g)\cap \ran(\tilde{F}_g)^\perp=\{0\}$ and (ii) $\tilde{F}_g\tilde{E}_{g'}=0$ if $\{g,g'\}\in E(\overline{G})$. For~(i), since $\ran(\tilde{E}_g)$ is the orthogonal complement of $\ran(X_g)=\ran(E_g)\cap \ran(F_g)^\perp$ in $\ran(E_g)$, it is straightforward to see that $\ran(\tilde{E}_g)\cap\ran(\tilde{F}_g)^\perp=\{0\}$. For (ii), since $\ran(\tilde{E}_{g'})\subseteq \ran(E_{g'})\perp\ran(F_g)$ for every $\{g,g'\}\in E(\overline{G})$, we have $\tilde{F}_g\tilde{E}_{g'}=0$ as well.
Thus, $(\cH,\tilde{\cM},\{\tilde{E}_g\},\{\tilde{F}_g\},\tilde{\tau})$ is a $\lambda$-tracial subspace representation of $G$ in terms of~\cref{prop: first relaxation of tracial Haemers bound}.

For the reverse direction, let $(\cH,\cM, \{E_g\}, \tau)$ be a feasible tuple for~\cref{def: tracial subspace rep}. Let $F_g$ be the projection onto $(\sum_{g' \in N_{\overline G}(g)} \ran(E_{g'}))^\perp$. Then $F_g E_{g'} = 0$ for all $\{g,g'\} \in E(\overline G)$ and $X_g = E_g \wedge (I_\cM-F_g) = 0$ for all $g \in V(G)$. Since $\tau(E_g)-\tau(X_g)=\tau(E_g)=1/\lambda$, we get a feasible solution of~\cref{prop: second relaxation of tracial Haemers bound} with value $\lambda$.
\end{proof}

\begin{proposition}\label{prop: third relaxation of tracial Haemers bound}
Let $G$ be a graph. Then $\cH_{tr}(G)$ equals the infimum of $\lambda$ for which there exist a Hilbert space $\cH$, a von Neumann algebra $\cM\subseteq B(\cH)$ containing projections $E_g$ and $F_g$ for all $g\in V(G)$, with a normal tracial state $\tau:\cM \to \C$, such that
\begin{enumerate}[label=\upshape(\roman*)]
\item $\norm{E_g(I_\cM-F_{g})}<1$ for any $g\in V(G)$,
\item $F_gE_{g'}=0$ for any $\{g,g'\}\in E(\overline{G})$,
\item $\frac{\tau(E_g)}{\tau(E)-\tau(X)}\geq \frac{1}{\lambda}$ for all $g\in V(G)$, where $E$ and $F$ are the projections onto the closures of the subspaces $\sum_{g\in V(G)}\ran(E_g)$ and $\sum_{g\in V(G)}\ran(F_g)$, respectively and $X$ is the projection onto $\ran(E)\cap\ran(F)^\perp$.
\end{enumerate}
\end{proposition}
\begin{proof}
Denote the infimum defined in the proposition  by $\hat \cH_{tr}(G)$. We first show $\hat \cH_{tr}(G)\leq \cH_{tr}(G)$. Let $(\cH,\cM,\{E_g\},\{F_g\},\tau)$ be a $\lambda$-tracial projection representation of $G$.
Note that the projections $E$ and $X$ defined as in (iii) above commute. In particular, $E-X$ is the projection onto the orthogonal complement of $\ran(E)\cap\ran(F)^\perp$ in $\ran(E)$ and therefore $\tau(E-X)\leq \tau(I_\cM)=1$. This shows that $\frac{\tau(E_g)}{\tau(E)-\tau(X)}\geq \tau(E_g)\geq \frac{1}{\lambda}$ for every $g\in V(G)$. Then it is clear that the same tuple $(\cH,\cM,\{E_g\},\{F_g\},\tau)$ satisfies the conditions of \cref{prop: third relaxation of tracial Haemers bound} with value $\lambda$. The inequality $\hat \cH_{tr}(G)\leq \cH_{tr}(G)$ follows by taking the infimum over all such $\lambda$.

Conversely, for every feasible solution $(\cH,\cM,\{E_g\},\{F_g\},\tau)$ of $\hat{\cH}_{tr}(G)$ with value $\lambda$, we construct a $\lambda$-tracial subspace representation of $G$ in terms of~\cref{prop: second relaxation of tracial Haemers bound}. Let $\tilde{E}_g$ be the left support projection of $(E-X)E_g$ and $\tilde{F}_g$ be the left support projection of $ F_g(E-X)$ for every $g\in V(G)$. Note that $\cl(\ran((E-X)E_g))\subseteq \ran(E-X)$ and $\ker(F_g(E-X))^\perp\subseteq \ran(E-X)$.\footnote{For the second, we know that $\ker(E-X)\subseteq\ker(F_g(E-X))$, thus $\ker(F_g(E-X))^\perp\subseteq \ker(E-X)^\perp=\ran(E-X)$.} Thus $\tilde{E}_g$ and $\tilde{F}_g$ commute with $E-X$ and we have $\tilde{E}_g(E-X)=\tilde{E}_g$ and $\tilde{F}_g(E-X)=\tilde{F}_g$. Let $\tilde{\cM}$ be the von Neumann algebra generated by $\{\tilde{E}_g\}$, $\{\tilde{F}_g\}$ and $E-X$ (as the identity element). $\tilde{\cM}$ is a subset of $\cM$ because $\tilde{E_g},\tilde{F_g}\in \cM$ are support projections. Define the tracial state $\tilde{\tau}:\tilde{\cM}\to\C$ by $\tilde{\tau}(Y)=\frac{\tau(Y)}{\tau(E-X)}$. We verify that $(\cH,\tilde{\cM},\{\tilde{E}_g\},\{\tilde{F}_g\},\tilde{\tau})$ is a $\lambda$-tracial subspace representation in terms of~\cref{prop: second relaxation of tracial Haemers bound}.

The key is to show that $F_g(E-X)E_{g'}=F_gE_{g'}$ for every $g,g'\in V(G)$. Indeed, $EE_{g'}=E_{g'}$ for any $g'\in V(G)$ since $E_{g'}\subseteq \bigvee_{g\in V(G)} E_g=E $. Also, $F_g X=0$ because $X$ is the projection onto $\ran(E)\cap (\sum_{g\in V(G)} \ran(F_g))^\perp$, which is a subspace of $\ran(F_g)^\perp$ for any $g\in V(G)$. Thus, $F_g(E-X)E_{g'}= F_g(EE_{g'}) - (FX)E_{g'} = F_gE_{g'}$.

Next, we show that $\ran(\tilde{E}_g)\cap\ran(\tilde{F}_g)^\perp=\{0\}$ for every $g\in V(G)$. Let
\[
v\in \ran(\tilde{E}_g)\cap\ran(\tilde{F}_g)^\perp=\cl(\ran((E-X)E_g))\cap\ker(F_g(E-X)).
\]
Then $\displaystyle v=\lim_{i\to\infty} (E-X)E_gu_i$ for some $\{u_i\}_{i\in\N} \subseteq \cH$ and $F_g(E-X)v=0$. Note that
\[
\lim_{i\to\infty}F_gE_gu_i=\lim_{i\to\infty}(F_g(E-X))((E-X)E_gu_i)=F_g(E-X)v=0.
\]
Using condition (i) (i.e.~$\norm{E_g(I_\cM-F_g)}=\delta_g<1$) and the triangle inequality, we obtain
\[
\norm{F_gE_gu_i}\ge \norm{E_g u_i}-\norm{(I_\cM-F_g)E_g u_i}\ge \norm{E_g u_i}-\norm{E_g(I_\cM-F_g)}\norm{E_gu_i}\ge (1-\delta_g)\norm{E_gu_i}.
\]
Therefore
\[
\lim_{i\to \infty}\norm{E_gu_i}\le (1-\delta_g)^{-1}\lim_{i\to \infty}\norm{F_gE_gu_i} =0
\]
which implies $\displaystyle v=\lim_{i\to \infty}(E-X)E_gu_i=0$.

We then verify that $\tilde{F}_g\tilde{E}_{g'}=0$ for all $\{g,g'\}\in E(\overline{G})$. Note that $\tilde{F}_g$ is the right support projection of $F_g(E-X)$, we have $\tilde{F}_g ((E-X) E_{g'})=0$ since $F_g(E-X)E_{g'}=F_gE_{g'}=0$ for $\{g,g'\}\in E(\overline{G})$. This implies $\tilde{F}_g\tilde{E}_{g'}=0$ since $\tilde{E}_{g'}$ is the left support projection of $(E-X)E_{g'}$.

Finally we verify that $\tilde{\tau}(\tilde{E}_g)\geq \frac{1}{\lambda}$. We prove that $\tau(\tilde{E}_g)=\tau(E_g)$ for all $g\in V(G)$ and hence
\[
\tilde{\tau}(\tilde{E}_g)=\frac{\tau(\tilde{E}_g)}{\tau(E-X)}=\frac{\tau(E_g)}{\tau(E-X)}\geq\frac{1}{\lambda}.
\]
Using \cref{thm:rank} (ii) (where we take $P=E_g$ and $Q=E-X$), it suffices to show that $\ran(E_g)\cap\ran(E-X)^\perp=\{0\}$. Let $v\in \ran(E_g)\cap\ran(E-X)^\perp$. Then we have $(E-X)v=0$ and $v=E_gv=Ev$ because $v\in \ran(E_g)\subseteq \ran(E)$, therefore $Xv=v$. This means that $v\in\ran(X)=\ran(E)\cap\ran(F)^\perp$. On the other hand, note that $\ran(F)^\perp\subseteq\ran(F_g)^\perp$. We conclude that $v\in\ran(E_g)\cap\ran(F_g)^\perp$, which by assumption $\ran(E_g) \cap \ran(F_g)^\perp=\{0\}$. That completes the proof.
\end{proof}

We use \cref{prop: first relaxation of tracial Haemers bound 2,prop: second relaxation of tracial Haemers bound,prop: third relaxation of tracial Haemers bound} to prove super-multiplicativity of tracial Haemers bound.

\begin{proposition}\label{prop: tracial Haemers supermultiplicativity}
For any graphs $G$ and $H$, $\cH_{tr}(G\boxtimes H)\geq\cH_{tr}(G)\cH_{tr}(H)$.
\end{proposition}
\begin{proof}
Let $(\cH, \cM, \{E_{g,h}\}, \{F_{g,h}\}, \tau)$ be a $\lambda$-tracial projection representation of $G\boxtimes H$ as in \cref{def: tracial projection rep}. Let $E_g$ and $F_g$ be the projections onto $\cl(\sum_{h\in V(H)} \ran(E_{g,h}))$ and $\cl(\sum_{h\in V(H)} \ran(F_{g,h}))$ for any fixed $g$, respectively. By construction we have that $E_g$ and $F_g$ belong to $\cM$. Note that $F_{g,h}E_{g',h'}=0$ for any $\{g,g'\}\in E(\overline{G})$ and any $h,h'\in V(H)$ because such $\{(g,h),(g',h')\}\in E(\overline{G\boxtimes H})$. This implies that $F_gE_{g'}=0$ for any $\{g,g'\}\in E(\overline{G})$. Let $X_g$ be the projection onto the subspace $\ran(E_g)\cap \ran(F_g)^\perp$. Let $\mu^{-1}=\min\{\tau(E_g)-\tau(X_g):~g\in V(G)\}$. We obtained a feasible solution of $\cH_{tr}(G)$ with value $\mu$ as given in \cref{prop: second relaxation of tracial Haemers bound}. This shows that $\cH_{tr}(G)\leq\mu$.

We now construct a $\lambda/\mu$-tracial projection representation of $H$ in terms of~\cref{prop: third relaxation of tracial Haemers bound}. Let $g_0\in V(G)$ be such that $\mu^{-1}=\tau(E_{g_0})-\tau(X_{g_0})$. Note that for $\{h,h'\}\in E(\overline{H})$, $E_{g_0,h}F_{g_0,h'}=0$ and for $h\in V(H)$, $\norm{E_{g_0,h}(I_\cM-F_{g_0,h})}<1$. Thus $(\cH,\cM, \{E_{g_0,h}\}, \{F_{g_0,h}\}, \tau)$ is a feasible solution of $\cH_{tr}(H)$ where we use the formulation given in \cref{prop: third relaxation of tracial Haemers bound}. To compute its value, we note that for every $h\in V(H)$ we have
\[
\frac{\tau(E_{g_0,h})}{\tau(E_{g_0})-\tau(X_{g_0})}=\mu\tau(E_{g_0,h})\geq\frac{\mu}{\lambda}.
\]

We conclude that  for any $\lambda \geq \cH_{tr}(G \boxtimes H)$, we have $\cH_{tr}(G)\cH_{tr}(H) \leq \mu \cdot \frac{\lambda}\mu = \lambda$. Taking the infimum over such $\lambda$ gives that $\cH_{tr}(G\boxtimes H)\geq\cH_{tr}(G)\cH_{tr}(H)$.
\end{proof}

As we have shown that the tracial Haemers bound is multiplicative (\cref{prop: tracial Haemers submultiplicative,prop: tracial Haemers supermultiplicativity}), additive (\cref{prop: subadditive,prop: superadditive}), normalized (\cref{ex: tracial Haemers empty graph}) and monotone with respect to $\leq_{qc}$ (\cref{prop: monotone tracial Haemers}), we conclude the following:
\begin{corollary}
The tracial Haemers bound is an element of $\bX(\cG,\leq_{qc})$.
\end{corollary}

\subsection{Comparing the fractional and tracial Haemers bounds}
A natural question to ask is whether the fractional and tracial Haemers bounds can be strictly separated. Here we show a connection between this question and Connes' embedding conjecture. Connes' embedding conjecture (CEC), first proposed by Connes in his famous work~\cite{Connes76}, is a fundamental problem in the field of operator algebras (we give a precise formulation below). Over decades CEC %Connes' embedding conjecture 
has been shown to be equivalent to many important conjectures in various branches of operator algebras, mathematics and computer science. A disproof of CEC %Connes' embedding conjecture 
is announced in the preprint~\cite{JNVWY}.
% (at the time of writing, the manuscript is still under peer-review)
In~\cite{JNVWY} the authors take a computational complexity approach based on an equivalent formulation of CEC %Connes' embedding conjecture 
in terms of non-local games. Given the importance of CEC, %Connes' embedding conjecture, 
finding an alternative proof of its failure is of great interest. We shall show that finding a graph $G$ for which $\cH_{tr}(G) \neq \cH_f(G)$ would provide such an alternative proof.
%\SG{We show in \cref{thrm:connes Haemers} that a separation between these bounds proves that Connes' embedding conjecture is false.} %Note that a disproof of Connes' embedding conjecture was proposed in~\cite{JNVWY} (at the time of writing, the manuscript is still under peer-review).
%We prove that such a separation implies that Connes' embedding conjecture is false.
This is similar to~\cite[Corollary 3.10]{synchronous} which showed the analogous statement for a separation between projective rank and tracial rank.

%Connes' embedding conjecture, first proposed by Connes in his famous work~\cite{Connes76}, is a long-standing open problem in the field of operator algebra.\footnote{A disproof of Connes' embedding conjecture was proposed in~\cite{JNVWY}. However, the manuscript is still under peer-review. } Over decades Connes' embedding conjecture has been shown to be equivalent to many important conjectures in various branches of operator algebras, mathematics and computer science.

CEC %Connes' embedding conjecture 
states that every von Neumann algebra $(\cM,\tau)$ with a tracial state $\tau$ is embeddable into an ultra-product $(M_n)^\omega$ of matrix algebra $M_n$. In terms of Voiculescu's microstate formulation~\cite{Voiculescu02}, $(\cM,\tau)$ is embeddable into $(M_n)^\omega$ if and only if for every finite family $\{x_1, \dots, x_n\}\subseteq\cM$ of self-adjoint elements, any $\eps >0$ and integer $N\ge 1$, there is some $k \in \N$ and $k\times k$ self-adjoint matrices $X_1, \cdots , X_n\in M_k$ such that for all $p\le N$ and all $i_1, \cdots, i_p \in [n]$,
\[|\mathrm{tr}(X_{i_1}X_{i_2} \cdots X_{i_p})-\tau(x_{i_1}x_{i_2} \cdots x_{i_p})|<\eps \ .\]
Here $\tr=\frac{1}{k}\text{Tr}$ is the normalized trace. We say that $(X_1, \cdots , X_n)$ is an $(N,\eps)$-microstate of $(x_1, \cdots, x_n)$. We say a tuple $(x_1, \cdots, x_n)$ has a matricial microstate if it admits $(N,\eps)$-microstate for all $N\ge 1$ and $\eps >0$. The following equivalent formulation is natural.

\begin{proposition}\label{prop:CEC}Let $\cM$ be a von Neumann algebra with a tracial state $\tau$. The followings are equivalent.
\begin{enumerate}[label=\upshape(\roman*)]
\item Every finite family $\{x_1,\dots,x_n\}$ of self-adjoint elements in $(\cM,\tau)$ admits a matricial microstate.
\item Every finite family $\{e_1,\dots,e_n\}$ of {\bf projections} in $(\cM,\tau)$ admits a matricial microstate of {\bf projections}, i.e., for all $N\ge 1$ and $\eps>0$, there is some $k \in \N$ and $k\times k$ {\bf projections} $E_1, \cdots , E_n\in M_k$ such that for all $p\le N$ and all $i_1, \cdots, i_p \in [n]$,
\[|\mathrm{tr}(E_{i_1}E_{i_2} \cdots E_{i_p})-\tau(e_{i_1}e_{i_2} \cdots e_{i_p})|<\eps \ .\]
\end{enumerate}
\end{proposition}
\begin{proof}
We first show (ii) $\Rightarrow$ (i). Let $\{x_1, \dots, x_n\}$ be a family of self-adjoint elements. Let $M=\max_{i}\norm{x_i}$.
For any $\eps>0$ and positive integer $N$, there exist an integer $L$ and self-adjoint elements $\tilde{x}_i=\sum_{j=1}^L \lambda_{i,j}e_{i,j}$ for every $i\in[n]$ of finite spectrum such that $\norm{x_i-\tilde{x}_i}\le \frac{\eps}{NM^{N-1}}$, where the $e_{i,j}$'s are projections. Indeed, let $L=\lceil\frac{2M}{\frac{\eps}{NM^{N-1}}}\rceil=\lceil\frac{2NM^N}{\eps}\rceil$, where $\lceil\cdot\rceil$ is the ceiling function. We can choose $\tilde{x}_i=f(x_i)$, where $f:[-M,M]\to \mathbb{R}$ is the step function:
\[ 
f(t)=\frac{(2\ell-L)M}{L}\ , \ \text{ for }t\in [\frac{(2\ell-L)M}{L},\frac{(2(\ell+1)-L)M}{L})\text{ and }\ell\in[L]\cup\{0\},
\]
and $e_{i,j}$ is the spectral projection of $x_i$ on the interval $[\frac{(2\ell-L)M}{L}, \frac{(2(\ell+1)-L)M}{L})$. Then for all $p \leq N$ and $i_1,\ldots, i_p \in [n]$ we have
\[|\tau(x_{i_1}x_{i_2} \cdots x_{i_p})-\tau(\tilde{x}_{i_1}\tilde{x}_{i_2} \cdots \tilde{x}_{i_p})|\le \norm{x_{i_1}x_{i_2} \cdots x_{i_p}-\tilde{x}_{i_1}\tilde{x}_{i_2} \cdots \tilde{x}_{i_p}}{}\le\frac{\eps}{NM^{N-1}}\cdot M^{N-1}\cdot N=\eps . \]
By the condition (ii), there exist projection $\{E_{i,j}\}_{1\le i\le n, 1\le j\le m} \subseteq M_k$ such that for all $i_1, \cdots, i_p \in [n]$ and $j_1, \cdots, j_p \in [m]$,
\[ |\text{tr}(E_{i_1,j_1}E_{i_2,j_2} \cdots E_{i_p,j_p})-\tau(e_{i_1,j_1}e_{i_2,j_2} \cdots e_{i_p,j_p})|<\frac{\eps}{m^NM^{N}} .\]
Take $X_{i}=\sum_{j}\lambda_{i,j}E_{i,j}$ as self-adjoint elements in $M_k$. We have
\begin{align*}
&|\text{tr}_k(X_{i_1}X_{i_2} \cdots X_{i_p})-\tau(x_{i_1}x_{i_2} \cdots x_{i_p})|\\ \le& |\text{tr}_k(X_{i_1}X_{i_2} \cdots X_{i_p})-\tau(\tilde{x}_{i_1}\tilde{x}_{i_2} \cdots \tilde{x}_{i_p})|+|\tau(x_{i_1}x_{i_2} \cdots x_{i_p})-\tau(\tilde{x}_{i_1}\tilde{x}_{i_2} \cdots \tilde{x}_{i_p})|
\\ \le& \sum_{1\le j_1,j_2,\cdots ,j_p\le m}|\lambda_{i_{1},j_{1}}\lambda_{i_{2},j_{2}}\cdots\lambda_{i_{p},j_{p}}|\cdot |\text{tr}(E_{i_1,j_1}E_{i_2,j_2} \cdots E_{i_p,j_p})-\tau(e_{i_1,j_1}e_{i_2,j_2} \cdots e_{i_p,j_p})| +\eps
\\  \le& \sum_{1\le j_1,j_2,\cdots ,j_p\le m}M^p \frac{\eps}{(mM)^{N}} +\eps
\\  \le &\eps +\eps=2\eps
\end{align*}
Since $\eps$ is arbitrary, this proves (ii)$\Rightarrow$(i).

(i)$\Rightarrow$(ii) essentially follows from~\cite[Proposition 3.6]{synchronous}. Indeed, consider a family of projections $\{e_1,\cdots,e_n\}\subseteq \cM$. Let $C^*(e_1,\cdots,e_n)$ be the $C^*$-algebra generated by $\{e_1,\dots,e_n\}$ and let $C^*(*_n\mathbb{Z}_2)$ be the universal $C^*$-algebra of $n$-projections (i.e.~the full group $C^*$-algebra of $n$-fold free product of $\mathbb{Z}_2$, see e.g.~\cite{blackadar2006operator}). Denote $a_i$ as the projection generator of $ C^*(*_n\mathbb{Z}_2)$.
 By the universality of $C^*(*_n\mathbb{Z}_2)$, there exists a $*$-homomorphism $\pi:C^*(*_n\mathbb{Z}_2)\to C^*(e_1,\cdots,e_n)$ such that $\pi(a_i)=e_i$ for all $i$. Moreover, $\tau\circ \pi$ induces a tracial state on $C^*(*_n\mathbb{Z}_2)$ such that for any ${i_1},{i_2},\cdots,{i_p}\in \{1,\cdots, n\}$
\[\tau(e_{i_1}e_{i_2} \cdots e_{i_p})=\tau\circ \pi(a_{i_1}a_{i_2} \cdots a_{i_p}).\]
By \cite[Proposition 3.6]{synchronous},  for any $(N,\eps)$, we have a $*$-homomorphism $\pi_k:C^*(*_n\mathbb{Z}_2)\to M_k$ such that
\[|\tau\circ \pi(a_{i_1}a_{i_2} \cdots a_{i_p})- \tr \circ\pi_k(a_{i_1}a_{i_2} \cdots a_{i_p})|\le \eps ,\]
for all $p\le N$ and all $i_1, \cdots, i_p \in [n]$. Note that
\begin{align*}&\tau\circ \pi(a_{i_1}a_{i_2} \cdots a_{i_p})=\tau\Big(\pi(a_{i_1})\pi(a_{i_2}) \cdots \pi(a_{i_p})\Big)
=\tau(e_{i_1}e_{i_2} \cdots e_{i_p})\ \\
&\tr \circ\pi_k(a_{i_1}a_{i_2} \cdots a_{i_p})=\tr \Big(\pi_k(a_{i_1})\pi_k(a_{i_2}) \cdots \pi_k(a_{i_p})\Big)\ ,
\end{align*}
and $\pi_k(a_{1}),\pi_k(a_{2}), \cdots, \pi_k(a_{n})$ are projections in $M_k$ (since $\pi_k$ is $*$-homomorphism). Therefore, $(\pi_k(a_{1}),\pi_k(a_{2}), \cdots, \pi_k(a_{n}))$ is an $(N,\eps)$-micro state of $(e_1,e_2,\cdots e_n)$.
\end{proof}

We prove that if $\mathcal{H}_{tr}(G)<\mathcal{H}_f(G)$ for some graph $G$,
then every von Neumann algebra admitting a $\lambda$-tracial projection representation of $G$ with $\lambda<\mathcal{H}_f(G)$ does not satisfy the ii) in above proposition. We need the following formulation of a (finite-dimensional) subspace representation, which can be seen as the finite-dimensional version of~\cref{prop: second relaxation of tracial Haemers bound}. We show that this formulation is equivalent to \cref{def: fractional Haemers}  in~\cref{sec: equiv. subspace rep.}.
\begin{restatable}{proposition}{eqsubrep}\label{def: fractional Haemers 2}
$G$ has a $(d,r)$-subspace representation (over $\C$) if and only if there exist subspaces $S_g,T_g\subseteq \C^d$ for all $g\in V(G)$ satisfying
\begin{enumerate}[label=\upshape(\roman*)]
\item $\dim(S_g)-\dim(S_g\cap T_g^\perp)\geq r$ for all $g\in V(G)$;
\item $S_{g'}\subseteq T_g^\perp~\forall\ \{g,g'\}\in E(\overline{G})$.
\end{enumerate}
\end{restatable}

We also need the following lemma from Dykema and Paulsen (and its corollary). Here $\norm{x}_{\tau,2}=\sqrt{\tau(x^*x)}$ denotes the $L_2$-norm with respect to the tracial state $\tau$.
\begin{lemma}[{\cite[Lemma 3.12]{synchronous}}]\label{lemma:orth}Let $G$ be a graph and $\cM$ be a von Neumann algebra equipped with a tracial state $\tau$. For every $\eps > 0$, there
is a $\delta > 0$ such that if $\{E_g\}_{g\in V(G)}$ are projections in $\cM$ satisfying $\tau(E_gE_{g'})<\delta$ for all $\{g,g'\}\in E(G)$, then there exist projections $\{\tilde{E}_g\}_{g\in V(G)}$ in $\cM$ satisfying $\tilde{E}_g\tilde{E}_{g'}=0$ for all  $\{g,g'\}\in E(G)$ and $\|E_g-\tilde{E}_{g}\|_{\tau,2} \le \eps$ for every $g\in V(G)$.
\end{lemma}
\begin{corollary}\label{cor:orth} Let $G$ be a graph. For every $\eps > 0$, there
is a $\delta > 0$ such that if $\{E_g\}_{g\in V(G)}$ and $\{F_g\}_{g \in V(G)}$ are projections in a von Neumann algebra $\cM$ satisfying
\[
\tau(F_g E_{g'})<\delta \qquad \text{ for all }\{g,g'\}\in E(G),
\]
then there exist projections $\{\tilde{E}_g\}_{g\in V(G)}$ and $\{\tilde{F}_g\}_{g\in V(G)}$ in $\cM$ satisfying
\[
\tilde{F}_g\tilde{E}_{g'}=0 \quad \text{for all  } \{g,g'\}\in E(G), \quad \text{and } \ \|E_g-\tilde{E}_{g}\|_{\tau,2}, \, \|F_g - \tilde F_g\|_{\tau,2} \leq \eps \quad \text{for every }g\in V(G).
\]
\end{corollary}
\begin{proof}
Let $G$ be a graph and let $H$ be the bipartite graph constructed from $G$ whose left and right vertex sets are two copies of the vertices of $G$. There is an edge between a vertex $g$ in the left copy and a vertex $g'$ in the right copy if and only if $\{g,g'\} \in E(G)$. For a vertex $g \in V(G)$ we assign $E_g$ to the left copy of $g$ and $F_g$ to the right copy of $g$. The corollary then follows from applying \cref{lemma:orth} to $H$.
\end{proof}

We then show the main result of this section.
\begin{theorem} \label{thrm:connes Haemers}
%If Connes' embedding conjecture is true, then $\mathcal{H}_{tr}(G)=\mathcal{H}_f(G)$ for every graph~$G$.
If $\mathcal{H}_{tr}(G)<\mathcal{H}_f(G)$ for a graph $G$, then the von Neumann algebra $(\cM,\tau)$ in any $\lambda$-tracial projection representation of $G$ for $\lambda<\mathcal{H}_f(G)$ does not satisfy Connes' embedding conjecture. In particular, the projections $\{E_g,F_g\}_{g\in V(G)}$ in the $\lambda$-tracial projection representation (as given in~\cref{def: tracial projection rep}) do not admit a matricial microstate.
\end{theorem}
\begin{proof}
Let $\cM$ be a von Neumann algebra with a normal tracial state $\tau$. Assume that $\{e_g\}_{g\in V(G)}$ and $\{f_g\}_{g\in V(G)}$ are projections satisfying the conditions in \cref{def: tracial projection rep}, i.e.
\begin{enumerate}[label=\upshape(\roman*)]
\item $\norm{e_g(I_\cM-f_g)}{}<1$ for any $g\in V(G)$.
\item $f_ge_{g'}=0$ for any $\{g,g'\}\in E(\overline{G})$,
\item $\tau(e_g) \geq \frac{1}{\lambda}$ for all $g\in V(G)$.
\end{enumerate}
We show by contradiction that $\{e_g\}\cup\{f_g,I_\cM-f_g\}_{g\in V(G)}$ cannot have a matricial microstate. By condition (i), we have $\lim_{n\to\infty}\norm{(e_g(I_\cM-f_g)e_g)^n}{}=0$. Then for every $\eps>0$, there exists a finite integer $\ell$ such that for any $g\in V(G)$ and $n\ge \ell$,
\[\tau((e_g(I_\cM-f_g)e_g)^{n})\le \eps.\]
Take the $\delta$ corresponding to $\eps/(3\ell)$ in \cref{cor:orth}. Suppose
$\{e_g\}\cup\{f_g,I_\cM-f_g\}_{g\in V(G)}$ admits a $(3\ell,\eps_2)$-matrical microstate for $\eps_2:=\min\{\delta,\epsilon\}$. Then there exists $k\in\N$ and projections $E_g,F_g\subseteq M_{k}$ for $g\in V(G)$ satisfying
\begin{enumerate}[label=\upshape(\roman*)]
\item $\tr((E_g(I_k-F_g)E_g)^{\ell})\le \eps+\eps_2$ for each $g\in V(G)$,
\item $\tr(F_gE_{g'})\le \eps_2$ for any $\{g,g'\}\in E(\overline{G})$,
\item $\tr(E_{g}) \geq \frac{1}{\lambda}-\eps_2$ for all $g\in V(G)$.
\end{enumerate}
We now use \cref{cor:orth} (on the complement graph $\overline{G}$) to convert (ii) to an orthogonality condition. There exist projections (again of size $k\times k$) $\tilde{E}_{g},\tilde{F}_{g}$ for $g\in V(G)$ such that $\tilde{F}_g\tilde{E}_{g'}=0$ for any $\{g,g'\}\in E(\overline{G})$ and for all $g \in V(G)$ we have
\[\norm{E_g-\tilde{E}_{g}}_{\tr,2} \le \frac{\eps}{3\ell},~\norm{F_g-\tilde{F}_{g}}_{\tr,2} \le \frac{\eps}{3\ell}.\]

Then we have
\begin{align*}
\tr((\tilde{E}_g(I_k-\tilde{F}_g)\tilde{E}_g)^{\ell})&\le \tr((E_g(I_k-F_g)E_g)^{\ell})+\tr(|(E_g(I_k-F_g)E_g)^{\ell}-(\tilde{E}_g(I_k-\tilde{F}_g)\tilde{E}_g)^{\ell}|) \\& \le \eps + \eps_2+ \norm{(E_g(I_k-F_g)E_g)^{\ell}-(\tilde{E}_g(I_k-\tilde{F}_g)\tilde{E}_g)^{\ell}}_{\tr,2}
\\& \le \eps+\eps+3\ell\cdot \frac{\eps}{3\ell}=3\eps,
\end{align*}
where the second inequality uses $\tr(|A|)\leq\norm{A}_{\tr,2}$ for the normalized trace ``$\tr$'' and the third inequality uses the triangle inequality
$3\ell$ times and the fact that the norm of products of projections is at most~$1$.

This shows that for any $\eps>0$, there exist $\ell,k \in \N$ and $k\times k$ projections $\tilde{E}_g$, $\tilde{F}_g$ for all $g\in V(G)$ satisfying:
\begin{itemize}
\item $\tr((\tilde{E}_g(I_k-\tilde{F}_g)\tilde{E}_g)^{\ell})\leq 3\eps$ for any $g\in V(G)$,
\item $\tilde{F}_g\tilde{E}_{g'}=0$ for any $\{g,g'\}\in E(\overline{G})$.
\end{itemize}
Let $X_g=\tilde{E}_g\wedge (I_k-\tilde{F}_g)$ be the projection onto $\text{ran}(\tilde{E}_g)\cap \text{ran}(I_k-\tilde{F}_g)$. By monotonicity of the sequence
$(\tilde{E}_g(I_k-\tilde{F}_g)\tilde{E}_g)^{n}\to X_g$, we have
$\tr(X_g)\le \tr((\tilde{E}_g(I_k-\tilde{F}_g)\tilde{E}_g)^{\ell})\le 3\eps$. We show that this gives a feasible solution of $\cH_f(G)$ in terms of~\cref{def: fractional Haemers 2}. Let $S_g=\ran(\tilde{E}_g)$ and $T_g=\ran(\tilde{F}_g)$.
Note that
\[
\dim(S_g)-\dim(S_g\cap T_g^\perp)=k(\tr(\tilde{E}_g-X_g))\geq k(\tr(E_g)-\tr(|E_g-\tilde{E_g}|)-3\eps)\geq k(\frac{1}{\lambda}-\eps_2-\frac{\eps}{3\ell}-3\eps).
\]
Thus
\[
\cH_f(G)\leq \frac{1}{\frac{1}{\lambda}-5\eps}=\frac{\lambda}{1-5\lambda\eps}.
\]
Taking $\eps$ such that $\frac{\lambda}{1-5\lambda\eps}< \cH_f(G)$, we reach a contradiction.
\end{proof}

%Last year, Ji, Natarajan, Vidick, Wright and Yuen in a preprint \cite{JNVWY} announced a refutation of Connes' embedding conjecture. They take a computational complexity theoretic approach that is based on a formulation of Connes' embedding conjecture in terms of non-local games.
%\cref{thrm:connes Haemers} shows that separating the tracial Haemers bound $\mathcal{H}_{tr}(G)$ and fractional Haemers bound $\mathcal{H}_f(G)$ would provide an alternative way to refute Connes' embedding conjecture. It also illustrates that it is likely nontrivial to separate them.

\section{Inertia upper bound on the commuting quantum independence number}\label{sec: inertia}

We recall the infinite-dimensional version of the projective packing number.
\begin{definition}[Tracial packing number]
Let $G$ be a graph. A $\lambda$-tracial packing of $G$ is a $C^*$-algebra $\cA$ containing projections $\{E_{g}\}\subseteq\mathcal{A}$ and equipped with a tracial state $\tau$, such that $E_gE_{g'}=0$ for all $\{g,g'\}\in E(G)$ and $\lambda=\sum_{g\in V(G)}\tau(E_g)$. The tracial packing number of $G$, which we denote by $\al_{tr}(G)$, is defined as the supremum of $\lambda$ over all possible $\lambda$-tracial packing of $G$.
\end{definition}
Note that the value $\alpha_{tr}(G)$ in the above definition remains the same if we replace $C^*$-algebras by von Neumann algebras and replace tracial states by normal tracial states.
When restricting to finite-dimensional $C^*$-algebras $\cA$, the above definition gives the projective packing number $\al_p(G)$, which is an upper bound on $\alpha_q(G)$~\cite{robersonthesis}. Similarly, the tracial packing number $\alpha_{tr}(G)$ upper bounds $\alpha_{qc}(G)$ (cf.~\cite[Proposition 9]{gribling2018bounds}).
\begin{definition}[Inertia bound]
Let  $A\in M_s$ be a Hermitian matrix. A subspace $S \subseteq\mathbb{C}^s$ is \emph{totally isotropic with respect to $A$} if for any vector $v\in S$, $\langle v,Av\rangle=0$.
The inertia bound of $A$, denoted as $n(A)$, is the maximum dimension of totally isotropic subspaces of $A$.
\end{definition}
It was proved in \cite{WE18,qinertia} that for any weighted adjacency matrix $A$ of $G$, we have $\al_{p}(G)\le n(W)$. We show this is also true for the tracial packing number. We first prove the following characterization of the inertia bound. Note that the original definition can be viewed as taking  $\cM = \C$.

\begin{proposition}\label{prop:inertia}
Let $A\in M_s$ be a Hermitian matrix and let $\cM$ be a von Neumann algebra equipped with a normal tracial state $\tau$. Then
\[n(A)=\sup_{E} \{ \Tr\otimes \tau(E)\  | \ E\in M_s\otimes\cM \text{ projection s.t.~for any $0\le P\le E$, } \Tr\otimes \tau((A\otimes I_{\cM})P)=0\}.\]
\end{proposition}
\begin{proof}Let $S\subseteq \mathbb{C}^s$ be a totally isotropic subspace of $A$ and let $E_S$ be the projection onto $S$. Then the direction ``$\le$'' follows from the fact that $E_S\otimes I_{\cM}$ is feasible for the right hand side supremum.
For the other direction, recall from~\cite[Lemma 1]{WE18} that
\[ n(A)=\Tr(E_0)+\min\{ \Tr(E_+),\Tr(E_-)\}.\]
where $E_0$ (resp.~$E_+$ and $E_-$) is the projection onto the eigenspace of $A$ of zero spectrum (resp.~positive spectrum and negative spectrum). Without loss of generality, we assume $\rk(E_+)=\Tr(E_+)\ge \Tr(E_-)=\rk(E_-)$ and the smallest positive eigenvalue of $A$ is $\la>0$. Then $E_+\otimes I_{\cM}$  is the spectral projection of $A\otimes I_{\cM}$ on $[\la,\infty)$.
For any projection $E\in M_s\otimes\cM$ such that
\[ \Tr\otimes \tau((A\otimes I_{\cM})P)=0, \forall \  0\le P\le E,\]
we have $(E_+\otimes I_\cM)\wedge E=0$. Indeed, for any nonzero projection $P\le E_+\otimes I_\cM$ and $P\le E$,
\[ \Tr\otimes \tau((A\otimes I_{\cM})P)\ge \la\Tr\otimes \tau(P)>0=\Tr\otimes \tau((A\otimes I_{\cM})P)\ .\]
which is a contradiction.
Thus by \cref{thm:rank} (i),
\begin{align*}
\Tr\otimes\tau(E)=&\Tr\otimes\tau(E)+\Tr\otimes\tau(E_+\otimes I_\cM)-\Tr\otimes\tau(E_+\otimes I_\cM)
\\  = &\Tr\otimes\tau(E\vee (E_+\otimes I_\cM))+\Tr\otimes\tau(E\wedge (E_+\otimes I_\cM)) -\Tr\otimes\tau(E_+\otimes I_\cM)
\\  \le &\Tr\otimes\tau(I_s\otimes I_{\cM}) -\Tr\otimes\tau(E_+\otimes I_{\cM})=s-\Tr(E_+) = \Tr(E_0)+\Tr(E_-)=n(A),
\end{align*}
where the inequality uses the fact that $E\wedge (E_+\otimes I_\cM)=0$ and $E\vee (E_+\otimes I_\cM)$ is a projection on $M_s\otimes\cM$.
Taking the supremum over all such $E$ completes the proof.
\end{proof}

\begin{proposition}
Let $G$ be a graph and let $A$ be a weighted adjacency matrix of $G$. We have $\al_{qc}(G)\le\al_{tr}(G)\le n(A)$.
\end{proposition}
\begin{proof}
We only need to prove the last inequality. Let $(\cM,\{E_g\},\tau)$ be a tracial packing of $G$. Define the projection $E=\sum_{g\in V(G)}e_ge_g^*\otimes E_g\in M_n\otimes\cM$, where $n=|V(G)|$ and $\{e_g\}$ is a set of standard basis of $\C^{n}$.
Note that
\[E(A\otimes I_\cM)E=\sum_{g,g'\in V(G)} \langle e_g,Ae_{g'}\rangle e_ge_{g'}^*\otimes E_gE_{g'}=0,\]
since $\langle e_g,Ae_{g'}\rangle=0$ for all $\{g,g'\}\notin E(G)$ and $E_gE_{g'}=0$ for all $\{g,g'\}\in E(G)$. Then for any $0\le P\le E$,
\[ \Tr\otimes \tau((A\otimes I_\cM)P)=\Tr\otimes \tau((A\otimes I_\cM)EPE)=\Tr\otimes\tau(E(A\otimes I_\mathcal{\cM})EP)=0.\]
It follows from \cref{prop:inertia} that
\[\Tr\otimes \tau(E)=\sum_{g\in V(G)}\tau(E_g)\le n(A),\]
which completes the proof.
\end{proof}

\paragraph*{Acknowledgments.}
We thank Monique Laurent for many interesting and helpful discussions and comments, and we thank the anonymous reviewers for many helpful comments. Part of this work was done when SG and YL were postdoctoral researchers at Centrum Wiskunde \& Informatica (CWI) and QuSoft, the Netherlands, and when LG was a postdoctoral researcher at Zentrum Mathematik, Technische Universit\"{a}t M\"{u}nchen.
SG's research was partially supported by SIRTEQ-grant QuIPP. YL's research was partially supported by MEXT Quantum Leap Flagship Program (MEXT Q-LEAP) Grant Number JPMXS0120319794.

\bibliographystyle{alphaurl}
\bibliography{alls}

\newcommand{\etalchar}[1]{$^{#1}$}
\begin{thebibliography}{CLMW10}

\bibitem[BBG13]{Briet19227}
Jop Bri{\"e}t, Harry Buhrman, and Dion Gijswijt.
\newblock {Violating the Shannon capacity of metric graphs with entanglement}.
\newblock {\em Proceedings of the National Academy of Sciences},
  110(48):19227--19232, 2013.
\newblock \href {https://doi.org/10.1073/pnas.1203857110}
  {\path{doi:10.1073/pnas.1203857110}}.

\bibitem[BBL{\etalchar{+}}15]{briet2015entanglement}
Jop Bri{\"e}t, Harry Buhrman, Monique Laurent, Teresa Piovesan, and Giannicola
  Scarpa.
\newblock Entanglement-assisted zero-error source-channel coding.
\newblock {\em IEEE Transactions on Information Theory}, 61(2):1124--1138,
  2015.
\newblock \href {https://doi.org/10.1109/TIT.2014.2385080}
  {\path{doi:10.1109/TIT.2014.2385080}}.

\bibitem[BC19]{bukh2018fractional}
Boris Bukh and Christopher Cox.
\newblock {On a fractional version of Haemers’ bound}.
\newblock {\em IEEE Transactions on Information Theory}, 65(6):3340--3348, June
  2019.
\newblock \href {https://doi.org/10.1109/TIT.2018.2889108}
  {\path{doi:10.1109/TIT.2018.2889108}}.

\bibitem[Bei10]{Beigi2010}
Salman Beigi.
\newblock {Entanglement-assisted zero-error capacity is upper-bounded by the
  Lov\'asz $\ensuremath{\vartheta}$ function}.
\newblock {\em Physical Review A}, 82:010303, Jul 2010.
\newblock \href {https://doi.org/10.1103/PhysRevA.82.010303}
  {\path{doi:10.1103/PhysRevA.82.010303}}.

\bibitem[Bla06]{blackadar2006operator}
Bruce Blackadar.
\newblock {\em {Operator algebras: theory of $C^*$-algebras and von Neumann
  algebras}}, volume 122.
\newblock Springer Science \& Business Media, 2006.

\bibitem[Bla13]{blasiak2013graph}
Anna Blasiak.
\newblock {\em A graph-theoretic approach to network coding}.
\newblock PhD thesis, Cornell University, 2013.
\newblock URL:
  \url{https://ecommons.cornell.edu/bitstream/handle/1813/34147/ab675.pdf}.

\bibitem[BSST02]{1035117}
Charles~H. {Bennett}, Peter~W. {Shor}, John~A. {Smolin}, and Ashish~V.
  {Thapliyal}.
\newblock {Entanglement-assisted capacity of a quantum channel and the reverse
  Shannon theorem}.
\newblock {\em IEEE Transactions on Information Theory}, 48(10):2637--2655,
  2002.
\newblock \href {https://doi.org/10.1109/TIT.2002.802612}
  {\path{doi:10.1109/TIT.2002.802612}}.

\bibitem[CLMW10]{cubitt2010improving}
Toby~S. Cubitt, Debbie Leung, William Matthews, and Andreas Winter.
\newblock Improving zero-error classical communication with entanglement.
\newblock {\em Physical Review Letters}, 104:230503, Jun 2010.
\newblock \href {https://doi.org/10.1103/PhysRevLett.104.230503}
  {\path{doi:10.1103/PhysRevLett.104.230503}}.

\bibitem[CMN{\etalchar{+}}07]{Cameron07quantumchrom}
Peter~J. Cameron, Ashley Montanaro, Michael~W. Newman, Simone Severini, and
  Andreas Winter.
\newblock On the quantum chromatic number of a graph.
\newblock {\em Electronic Journal of Combinatorics}, 14(1), 2007.

\bibitem[CMR{\etalchar{+}}14]{cubitt2014bounds}
Toby~S. Cubitt, Laura Man{\v{c}}inska, David~E. Roberson, Simone Severini, Dan
  Stahlke, and Andreas Winter.
\newblock Bounds on entanglement-assisted source-channel coding via the
  lov{\'a}sz theta number and its variants.
\newblock {\em IEEE Transactions on Information Theory}, 60(11):7330--7344,
  2014.
\newblock \href {http://arxiv.org/abs/1310.7120} {\path{arXiv:1310.7120}},
  \href {https://doi.org/10.1109/TIT.2014.2349502}
  {\path{doi:10.1109/TIT.2014.2349502}}.

\bibitem[CMSS14]{chailloux_et_al:LIPIcs:2014:4807}
Andr{\'e} Chailloux, Laura Mancinska, Giannicola Scarpa, and Simone Severini.
\newblock {Graph-theoretical Bounds on the Entangled Value of Non-local Games}.
\newblock In Steven~T. Flammia and Aram~W. Harrow, editors, {\em 9th Conference
  on the Theory of Quantum Computation, Communication and Cryptography (TQC
  2014)}, volume~27 of {\em Leibniz International Proceedings in Informatics
  (LIPIcs)}, pages 67--75, Dagstuhl, Germany, 2014. Schloss
  Dagstuhl--Leibniz-Zentrum fuer Informatik.
\newblock URL: \url{http://drops.dagstuhl.de/opus/volltexte/2014/4807}, \href
  {https://doi.org/10.4230/LIPIcs.TQC.2014.67}
  {\path{doi:10.4230/LIPIcs.TQC.2014.67}}.

\bibitem[Con76]{Connes76}
Alain Connes.
\newblock {Classification of injective factors Cases $\text{II}_1$,
  $\text{II}_\infty$, $\text{III}_\lambda$, $\lambda\ne 1$}.
\newblock {\em Annals of Mathematics}, pages 73--115, 1976.
\newblock \href {https://doi.org/10.2307/1971057} {\path{doi:10.2307/1971057}}.

\bibitem[Cve73]{inertia}
D.M. Cvetkovi\'c.
\newblock Inequalities obtained on the basis of the spectrum of the graph.
\newblock {\em Studia Sci. Math. Hungar.}, 8:433--436, 1973.

\bibitem[DP16]{synchronous}
Kenneth~J. Dykema and Vern~I. Paulsen.
\newblock {Synchronous correlation matrices and Connes’ embedding
  conjecture}.
\newblock {\em Journal of Mathematical Physics}, 57(1):015214, 2016.
\newblock \href {https://doi.org/10.1063/1.4936751}
  {\path{doi:10.1063/1.4936751}}.

\bibitem[DSW13]{duan2013}
Runyao Duan, Simone Severini, and Andreas Winter.
\newblock {Zero-Error communication via quantum channels, noncommutative
  graphs, and a quantum Lov\'{a}sz number}.
\newblock {\em IEEE Transactions on Information Theory}, 59(2):1164--1174, Feb
  2013.
\newblock \href {https://doi.org/10.1109/TIT.2012.2221677}
  {\path{doi:10.1109/TIT.2012.2221677}}.

\bibitem[Fek23]{Fekete1923}
M.~Fekete.
\newblock {\"U}ber die verteilung der wurzeln bei gewissen algebraischen
  gleichungen mit ganzzahligen koeffizienten.
\newblock {\em Mathematische Zeitschrift}, 17(1):228--249, Dec 1923.
\newblock \href {https://doi.org/10.1007/BF01504345}
  {\path{doi:10.1007/BF01504345}}.

\bibitem[GdLL18]{gribling2018bounds}
Sander Gribling, David de~Laat, and Monique Laurent.
\newblock Bounds on entanglement dimensions and quantum graph parameters via
  noncommutative polynomial optimization.
\newblock {\em Mathematical programming}, 170(1):5--42, 2018.
\newblock \href {https://doi.org/10.1007/s10107-018-1287-z}
  {\path{doi:10.1007/s10107-018-1287-z}}.

\bibitem[Hae78]{Haemers1978}
Willem Haemers.
\newblock An upper bound for the {Shannon} capacity of a graph.
\newblock {\em Colloquia Mathematica Societatis J\'{a}nos Bolyai}, 25:267--272,
  1978.

\bibitem[Hae79]{haemers1979some}
Willem Haemers.
\newblock {On some problems of Lov{\'a}sz concerning the Shannon capacity of a
  graph}.
\newblock {\em IEEE Transactions on Information Theory}, 25(2):231--232, 1979.
\newblock \href {https://doi.org/10.1109/TIT.1979.1056027}
  {\path{doi:10.1109/TIT.1979.1056027}}.

\bibitem[Ji13]{qalpha}
Zhengfeng Ji.
\newblock Binary constraint system games and locally commutative reductions.
\newblock arXiv:1310,3794, 2013.
\newblock \href {http://arxiv.org/abs/1310,3794} {\path{arXiv:1310,3794}}.

\bibitem[JNV{\etalchar{+}}20]{JNVWY}
Zhengfeng {Ji}, Anand {Natarajan}, Thomas {Vidick}, John {Wright}, and Henry
  {Yuen}.
\newblock {$\text{MIP}^*$=RE}.
\newblock arXiv: 2001.04383, January 2020.
\newblock \href {http://arxiv.org/abs/2001.04383} {\path{arXiv:2001.04383}}.

\bibitem[Knu94]{knuth1994sandwich}
Donald~E. Knuth.
\newblock The sandwich theorem.
\newblock {\em The Electronic Journal of Combinatorics}, 1(1):1, 1994.

\bibitem[KR83]{kadison1986fundamentalsI}
Richard~V. Kadison and John~R. Ringrose.
\newblock {\em {Fundamentals of the theory of operator algebras. Volume I:
  Elementary theory}}.
\newblock Graduate Studies in Mathematics, 1983.

\bibitem[KR86]{kadison1986fundamentalsII}
Richard~V. Kadison and John~R. Ringrose.
\newblock {\em {Fundamentals of the theory of operator algebras. Volume II:
  Advanced theory}}.
\newblock Graduate Studies in Mathematics, 1986.

\bibitem[LMM{\etalchar{+}}12]{Leung2012}
Debbie Leung, Laura Man{\v{c}}inska, William Matthews, M{\={a}}ris Ozols, and
  Aidan Roy.
\newblock Entanglement can increase asymptotic rates of zero-error classical
  communication over classical channels.
\newblock {\em Communications in Mathematical Physics}, 311(1):97--111, Apr
  2012.
\newblock \href {https://doi.org/10.1007/s00220-012-1451-x}
  {\path{doi:10.1007/s00220-012-1451-x}}.

\bibitem[Lov75]{10.1016/0012-365X(75)90058-8}
L\'{a}szl\'{o} Lov\'{a}sz.
\newblock On the ratio of optimal integral and fractional covers.
\newblock {\em Discrete Mathematics}, 13(4):383–390, 1975.
\newblock \href {https://doi.org/10.1016/0012-365X(75)90058-8}
  {\path{doi:10.1016/0012-365X(75)90058-8}}.

\bibitem[Lov79]{lovasz1979shannon}
L{\'a}szl{\'o} Lov{\'a}sz.
\newblock {On the Shannon capacity of a graph}.
\newblock {\em IEEE Transactions on Information theory}, 25(1):1--7, 1979.
\newblock \href {https://doi.org/10.1109/TIT.1979.1055985}
  {\path{doi:10.1109/TIT.1979.1055985}}.

\bibitem[Lov86]{doi:10.1137/1.9781611970203}
L\'{a}szl\'{o} Lov\'{a}sz.
\newblock {\em An Algorithmic Theory of Numbers, Graphs and Convexity}.
\newblock Society for Industrial and Applied Mathematics, 1986.
\newblock \href {https://doi.org/10.1137/1.9781611970203}
  {\path{doi:10.1137/1.9781611970203}}.

\bibitem[LZ21]{Li2018quantum}
Yinan Li and Jeroen Zuiddam.
\newblock Quantum asymptotic spectra of graphs and non-commutative graphs, and
  quantum {S}hannon capacities.
\newblock {\em IEEE Transactions on Information Theory}, 67(1):416--432, 2021.
\newblock \href {https://doi.org/10.1109/TIT.2020.3032686}
  {\path{doi:10.1109/TIT.2020.3032686}}.

\bibitem[MR16]{manvcinska2016quantum}
Laura Man{\v{c}}inska and David~E. Roberson.
\newblock Quantum homomorphisms.
\newblock {\em Journal of Combinatorial Theory, Series B}, 118:228--267, 2016.
\newblock \href {https://doi.org/10.1016/j.jctb.2015.12.009}
  {\path{doi:10.1016/j.jctb.2015.12.009}}.

\bibitem[MRV16]{cj16-05}
Laura Man\v{c}inska, David~E. Roberson, and Antonios Varvisotis.
\newblock On deciding the existence of perfect entangled strategies for
  nonlocal games.
\newblock {\em Chicago Journal of Theoretical Computer Science}, 2016(5), April
  2016.

\bibitem[MSS13]{6466384}
Laura {Mančinska}, Giannicola {Scarpa}, and Simone {Severini}.
\newblock New separations in zero-error channel capacity through projective
  {Kochen–Specker} sets and quantum coloring.
\newblock {\em IEEE Transactions on Information Theory}, 59(6):4025--4032,
  2013.
\newblock \href {https://doi.org/10.1109/TIT.2013.2248031}
  {\path{doi:10.1109/TIT.2013.2248031}}.

\bibitem[OP16]{ORTIZ2016}
Carlos~M. Ortiz and Vern~I. Paulsen.
\newblock Quantum graph homomorphisms via operator systems.
\newblock {\em Linear Algebra and its Applications}, 497:23 -- 43, 2016.
\newblock \href {https://doi.org/https://doi.org/10.1016/j.laa.2016.02.019}
  {\path{doi:https://doi.org/10.1016/j.laa.2016.02.019}}.

\bibitem[PSS{\etalchar{+}}16]{PAULSEN2016}
Vern~I. Paulsen, Simone Severini, Daniel Stahlke, Ivan~G. Todorov, and Andreas
  Winter.
\newblock Estimating quantum chromatic numbers.
\newblock {\em Journal of Functional Analysis}, 270(6):2188 -- 2222, 2016.
\newblock \href {https://doi.org/https://doi.org/10.1016/j.jfa.2016.01.010}
  {\path{doi:https://doi.org/10.1016/j.jfa.2016.01.010}}.

\bibitem[PT15]{paulsen2015}
Vern~I. Paulsen and Ivan~G. Todorov.
\newblock Quantum chromatic numbers via operator systems.
\newblock {\em The Quarterly Journal of Mathematics}, 66(2):677--692, 2015.
\newblock \href {https://doi.org/10.1093/qmath/hav004}
  {\path{doi:10.1093/qmath/hav004}}.

\bibitem[Rob13]{robersonthesis}
David~E. Roberson.
\newblock {\em Variations on a Theme: Graph Homomorphisms}.
\newblock PhD thesis, University of Waterloo, 2013.
\newblock URL: \url{https://uwspace.uwaterloo.ca/handle/10012/7814}.

\bibitem[RS89]{twoprojections}
Iain Raeburn and Allan~M. Sinclair.
\newblock The ${C}^*$-algebra generated by two projections.
\newblock {\em Mathematica Scandinavica}, 65(2):278--290, 1989.
\newblock URL: \url{http://www.jstor.org/stable/24491975}.

\bibitem[Sch03]{schrijver2003combinatorial}
Alexander Schrijver.
\newblock {\em {Combinatorial optimization: Polyhedra and efficiency}},
  volume~24.
\newblock Springer Science \& Business Media, 2003.

\bibitem[Sha48]{https://doi.org/10.1002/j.1538-7305.1948.tb01338.x}
C.~E. Shannon.
\newblock A mathematical theory of communication.
\newblock {\em Bell System Technical Journal}, 27(3):379--423, 1948.
\newblock \href
  {https://doi.org/https://doi.org/10.1002/j.1538-7305.1948.tb01338.x}
  {\path{doi:https://doi.org/10.1002/j.1538-7305.1948.tb01338.x}}.

\bibitem[Sha56]{MR0089131}
Claude~E. Shannon.
\newblock The zero error capacity of a noisy channel.
\newblock {\em Institute of Radio Engineers, Transactions on Information
  Theory}, IT-2(September):8--19, 1956.
\newblock \href {https://doi.org/10.1109/TIT.1956.1056798}
  {\path{doi:10.1109/TIT.1956.1056798}}.

\bibitem[Str88]{strassen1988asymptotic}
Volker Strassen.
\newblock The asymptotic spectrum of tensors.
\newblock {\em Journal f\"ur die Reine und Angewandte Mathematik. [Crelle's
  Journal]}, 384:102--152, 1988.
\newblock \href {https://doi.org/10.1515/crll.1988.384.102}
  {\path{doi:10.1515/crll.1988.384.102}}.

\bibitem[Voi02]{Voiculescu02}
Dan Voiculescu.
\newblock Free entropy.
\newblock {\em Bulletin of the London Mathematical Society}, 34(3):257--278,
  2002.

\bibitem[WE18]{WE18}
Pawel Wocjan and Clive Elphick.
\newblock An inertial upper bound for the quantum independence number of a
  graph.
\newblock arXiv:1808.10820, 2018.
\newblock \href {http://arxiv.org/abs/1808.10820} {\path{arXiv:1808.10820}}.

\bibitem[WEA19]{qinertia}
Pawel Wocjan, Clive Elphick, and Aida Abiad.
\newblock Spectral upper bound on the quantum $k$-independence number of a
  graph.
\newblock arXiv:1910.07339, 2019.
\newblock \href {http://arxiv.org/abs/1910,07339} {\path{arXiv:1910,07339}}.

\bibitem[Zui18]{phd}
Jeroen Zuiddam.
\newblock {\em Asymptotic spectra, algebraic complexity and moment polytopes}.
\newblock PhD thesis, University of Amsterdam, 2018.
\newblock URL:
  \url{https://dare.uva.nl/search?identifier=9a8030e9-f708-4c95-9d50-f2a5919e75ed}.

\bibitem[Zui19]{zuiddam2018asymptotic}
Jeroen Zuiddam.
\newblock {The asymptotic spectrum of graphs and the Shannon capacity}.
\newblock {\em Combinatorica}, 39:1173--1184, 2019.
\newblock \href {https://doi.org/10.1007/s00493-019-3992-5}
  {\path{doi:10.1007/s00493-019-3992-5}}.

\end{thebibliography}

\appendix

\section{Relaxing the equality in \texorpdfstring{\cref{def: tracial subspace rep}}{Definition 4.2}} \label{App: inequality}
We first show a similar result as \cref{def: tracial rank}: in the definition of a $\lambda$-tracial subspace representation we may relax the equality in (ii) to an inequality $\geq \frac{1}{\lambda}$. The proof strategy is similar to the one for tracial rank in~\cite[Theorem 6.11]{PAULSEN2016}.
\begin{proposition}\label{prop: alternative definition of tracial Haemers bound}
Let $G$ be a graph. If $\cM\subseteq B(\cH)$ is a von Neumann algebra containing projections $E_g$ for all $g\in V(G)$, with a normal tracial state $\tau:\cM \to \C$, such that, for
\[
E_{N_{\overline{G}}(g)} = \bigvee_{g'\in N_{\overline{G}}(g)} E_{g'}=\text{sot-}\lim_{n\to\infty }\Big(\sum_{g'\in N_{\overline{G}}(g)} E_{g'}\Big)^{1/n} \text{ for } g \in V(G),
\]
we have
\begin{enumerate}[label=\upshape(\roman*)]
\item $\ran(E_g) \cap \ran(E_{N_{\overline{G}}(g)})=\{0\}$.
\item $\tau(E_g) \geq \frac{1}{\lambda}$ for all $g\in V(G)$
\end{enumerate}
then $G$ has a $\lambda$-tracial subspace representation.
\end{proposition}
\begin{proof}
Define the scalars  $c_g=\tau(E_g)$ for $g\in V(G)$. By (ii) we have $c_g\geq 1/\lambda$ for all $g\in V(G)$. Consider the von Neumann algebra $L_\infty(0,1)$ represented on the Hilbert space $L_2(0,1)$ with tracial state $\tau_{2}(f) = \int_{(0,1)} f(x) \, \mathrm d x$,
where $\mathrm d x$ is the Lebesgue measure. For a scalar $r\in(0,1)$, let $P_r=I_{(0,r)}$ be the characteristic function on the open interval $(0,1)$. Since $c_g \geq 1/\lambda$ for all $g \in V(G)$, we have $(\lambda c_g)^{-1} \leq 1$. For $g\in V(G)$, define $\widehat{E}_g=E_g \otimes P_{(\lambda c_g)^{-1}}$, which is a projection in $\cM\overline{\otimes} L_\infty(0,1)\subseteq B(\cH\otimes L_2(0,1))$. Similarly, let $\widehat E_{N_{\overline{G}}(g)}$ be the projection
\[
\widehat E_{N_{\overline{G}}(g)}:=\text{sot-}\lim_{n\to\infty}\Big(\sum_{g'\in N_{\overline{G}}(g)} \widehat E_{g'}\Big)^{1/n} \text{ for } g \in V(G).
\]
Let $\tilde{\cM}\subseteq \cM\overline{\otimes} L_\infty(0,1)\subseteq B(\cH\otimes L^2(0,1))$ be the von Neumann algebra generated by the set of projections $\{\widehat{E}_g:~g\in V(G)\}$, equipped with the tracial state $\widehat{\tau} =\tau \otimes \tau_2$.
Then, for $g\in V(G)$, we have
\[
\widehat{\tau}(\widehat{E}_g)=\tau(E_g) \tau_2(P_{(\lambda c_g)^{-1}}) =c_g \cdot \frac{1}{\lambda c_g} = \frac{1}{\lambda}.
\]
We are only left to prove $\ran(\widehat{E}_g) \cap \ran(\widehat E_{N_{\overline{G}}(g)})=\{0\}$. Let $x\in\ran(\widehat{E}_g) \cap \ran(\widehat E_{N_{\overline{G}}(g)})$. Then there exist $y,y^{(N)}_{g'}\in\cH\otimes L^2(0,1)$ for all $g'\in N_{\overline{G}}(g)$ and $N\in\N$ such that
\[
x=\widehat{E}_g y=(E_g\otimes P_{(\lambda c_g)^{-1}})y
\]
and
\[
\lim_{N\to\infty}\|(\sum_{g'\in N_{\overline{G}}(g)}\widehat{E}_{g'}y_{g'}^{(N)})-x\|=\lim_{N\to\infty}\|(\sum_{g'\in N_{\overline{G}}(g)}(E_{g'}\otimes P_{(\lambda c_{g'})^{-1}}) y_{g'}^{(N)})-x\|=0.
\]
Note that for any linear functional (dual vector) $\mathcal L\in (L^2(0,1))^*$, we have $(I_\cA\otimes \mathcal L)(x)\in\ran(E_g)$. Moreover, let $z_{g'}^{(N)}=(I_\cA\otimes \mathcal L)((E_{g'}\otimes P_{(\lambda c_g)^{-1}})y_{g'}^{(N)})\in\cH$ for every $g'\in N_{\overline{G}}(g)$, we have
\[
\lim_{N\to\infty}\|(\sum_{g'\in N_{\overline{G}}(g)}z_{g'}^{(N)})-(I_\cA\otimes \mathcal L)(x)\|\leq\| \mathcal L\|\cdot \lim_{N\to\infty}\|(\sum_{g'\in N_{\overline{G}}(g)}(E_{g'}\otimes P_{(\lambda c_g)^{-1}}) y_{g'}^{(N)})-x\|=0
\]
Thus
\[
(I_\cA\otimes \mathcal L)(x)\in\ran(E_g)\cap\cl(\sum_{g'\in N_{\overline{G}}(g)}\ran(E_{g'}))=\{0\},
\]
for any $\mathcal L\in L^2(0,1)^*$,
which implies that $x=0$. Thus we obtain a $\lambda$-tracial subspace representation of~$G$.
\end{proof}

\section{Equivalent formulation of the subspace representation of a graph}
\label{sec: equiv. subspace rep.}
We prove the following equivalent formulation of the subspace representation of a graph (cf.~\cref{def: fractional Haemers}).
\eqsubrep*
\begin{proof}
Let $\{S_g\}_{g\in V(G)}$ be a $(d,r)$-subspace representation of $G$ as given in~\cref{def: fractional Haemers}. Let $T_g=(\sum_{g'\in N_{\overline{G}}(g)}S_{g'})^\perp$. Then it is straightforward to see that $S_{g'}\subseteq T_g^\perp$ if $\{g,g'\}\in E(\overline{G})$ and $S_g\cap T_g^\perp=S_g\cap(\sum_{g'\in N_{\overline{G}}(g)}S_{g'})=\{0\}$. Thus, $\{S_g\}_{g\in V(G)}$ and $\{T_g\}_{g\in V(G)}$ satisfy the conditions in~\cref{def: fractional Haemers 2}.

On the other hand, let $\{S_g\}_{g\in V(G)}$ and $\{T_g\}_{g\in V(G)}$ satisfy the conditions in~\cref{def: fractional Haemers 2} with value $(d,r)$, we construct a $(d,r)$-subspace representation of $G$ as given in~\cref{def: fractional Haemers}. For $g\in V(G)$, let $\hat{S}_g$ be the orthogonal complement of $S_g\cap T_g^\perp$ in $S_g$. By (i), we know that $\dim(\hat{S}_g)\geq r$. Let $\tilde{S_g}$ be a dimension-$r$ subspace of $\hat{S_g}$.  We are left to show that $\tilde{S}_g\cap(\sum_{g'\in N_{\overline{G}}(g)}\tilde{S}_{g'})=\{0\}$. Note that
\[
(\sum_{g'\in N_{\overline{G}}(g)}\tilde{S}_{g'})\subseteq(\sum_{g'\in N_{\overline{G}}(g)}S_{g'})\subseteq T_g^\perp,
\]
where the first inclusion is due to the fact that $\tilde{S}_g\subseteq \hat{S}_g\subseteq S_g$ for any $g\in V(G)$ and the second inclusion is due to (ii). Since $\tilde{S}_g$ is contained in the orthogonal complement of $S_g\cap T_g^\perp$ in $S_g$, we have $\tilde{S}_g\cap T_g^\perp=\{0\}$ and it follows that $\tilde{S}_g\cap(\sum_{g'\in N_{\overline{G}}(g)}\tilde{S}_{g'})=\{0\}$. This shows that $\{\tilde{S}_g\}_{g\in V(G)}$ is a $(d,r)$-subspace representation of $G$.
\end{proof}

\end{document}